%% file: erfairness_arxiv.tex
\documentclass[conference]{IEEEtran}

\usepackage[utf8]{inputenc} 
\usepackage[T1]{fontenc}    
\usepackage{hyperref}       
\usepackage{url}            
\usepackage{booktabs}       
\usepackage{amsfonts}       
\usepackage{nicefrac}       
\usepackage{microtype}      
\usepackage{xcolor}         
\usepackage{graphicx}
\usepackage{amsthm}
\usepackage{amsmath}

\usepackage{amssymb}
\usepackage{siunitx}
\usepackage{stfloats}
\usepackage{subfigure}
\usepackage{hyperref}
\usepackage{enumerate}
\usepackage{ulem}
\usepackage{enumitem}
\usepackage{xcolor}

\usepackage{bm}
\usepackage{nicefrac}
\usepackage{booktabs}
\usepackage{array}
\usepackage{multirow}
\usepackage{threeparttable}
\usepackage{makecell}
\usepackage[procnumbered,ruled,vlined,linesnumbered]{algorithm2e}
\usepackage{caption}
\usepackage{cite}

\allowdisplaybreaks[3]

\def\sizeof#1{\left|#1  \right|}

\def\eps{\epsilon}

\def\abs#1{\left|#1  \right|}

\def\trace#1{\mathrm{Tr} \left(#1 \right)}
\def\norm#1{\left\| #1 \right\|}

\newtheorem{problem}{Problem}
\newtheorem{theorem}{Theorem}[section]

\newtheorem{lemma}[theorem]{Lemma}

\newtheorem{fact}[theorem]{Fact}

\newtheorem{definition}[theorem]{Definition}

\def\calG{\mathcal{G}}
\def\calK{\mathcal{K}}

\def\norm#1{\left\| #1 \right\|}

\def\kh#1{\left( #1 \right)}

\newcommand{\removelatexerror}{\let\@latex@error\@gobble}

\newcommand\LL{\bm{\mathit{L}}}

\def\trace#1{\mathrm{Tr} \left(#1 \right)}
\def\sizeof#1{\left|#1  \right|}

\def\aa{\pmb{\mathit{a}}}

\newcommand\XX{\boldsymbol{\mathit{X}}}
\newcommand\yy{\boldsymbol{\mathit{y}}}
\newcommand\zz{\boldsymbol{\mathit{z}}}
\newcommand\xx{\boldsymbol{\mathit{x}}}

\newcommand\bb{\boldsymbol{\mathit{b}}}
\newcommand\cc{\boldsymbol{\mathit{c}}}
\newcommand\dd{\boldsymbol{\mathit{d}}}
\newcommand\ee{\boldsymbol{\mathit{e}}}

\newcommand\rr{\boldsymbol{\mathit{r}}}
\newcommand\vvv{\boldsymbol{\mathit{v}}}

\newcommand\ww{\boldsymbol{\mathit{w}}}

\renewcommand\AA{\boldsymbol{\mathit{A}}}
\newcommand\BB{\boldsymbol{\mathit{B}}}
\newcommand\CC{\boldsymbol{\mathit{C}}}
\newcommand\JJ{\boldsymbol{\mathit{J}}}
\newcommand\DD{\boldsymbol{\mathit{D}}}

\newcommand\EE{\boldsymbol{\mathit{E}}}

\newcommand\ZZ{\boldsymbol{\mathit{Z}}}

\newcommand\QQ{\boldsymbol{\mathit{Q}}}
\newcommand\MM{\boldsymbol{\mathit{M}}}

\newcommand\ZZtil{\boldsymbol{\mathit{\widetilde{Z}}}}
\newcommand{\add}{E_{a}}
\newcommand{\Greedy}{\textsc{Exact}\xspace}
\newcommand{\Farthest}{\textsc{Farthest}\xspace}
\newcommand{\Fast}{\textsc{Fast}\xspace}
\newcommand{\Gradient}{\textsc{Gradient}\xspace}
\newcommand{\SDDMSolver}{\textsc{Solve}}
\newcommand{\MFI}{\textsc{FIAM}\xspace}
\def\trace#1{\mathrm{Tr} \left(#1 \right)}
\def\ldtrace#1{\mathrm{Tr} \big(#1 \big)}
\def\norm#1{\left\| #1 \right\|}

\DeclareMathOperator*{\argmin}{arg\,min}
\DeclareMathOperator*{\argmax}{arg\,max}

\DontPrintSemicolon
\SetKw{KwAnd}{and}
\SetFuncSty{textsc}
\SetKwInOut{Input}{Input\ \ \ \ }
\SetKwInOut{Output}{Output}

\ifCLASSOPTIONcompsoc
  \usepackage[nocompress]{cite}
\else
  \usepackage{cite}
\fi

\ifCLASSINFOpdf

\else

\fi

\begin{document}
\title{Promoting Fairness in Information Access within Social Networks}
\IEEEoverridecommandlockouts 

\author{\IEEEauthorblockN{Changan~Liu}
\IEEEauthorblockA{College of Computer Science and Artificial Intelligence\\
Fudan University\\
Shanghai, 200433, China\\
19110240031@fudan.edu.cn}\\
\IEEEauthorblockN{Ahad N.~Zehmakan}
\IEEEauthorblockA{School of Computing\\
Australian National University\\
Canberra, Australia\\
ahadn.zehmakan@anu.edu.au}
\and
\IEEEauthorblockN{Xiaotian~Zhou}
\IEEEauthorblockA{College of Computer Science and Artificial Intelligence\\
Fudan University\\
Shanghai, 200433, China\\
20210240043@fudan.edu.cn} \\
\IEEEauthorblockN{Zhongzhi~Zhang\thanks{Zhongzhi~Zhang is the corresponding author.}}
\IEEEauthorblockA{College of Computer Science and Artificial Intelligence\\
Fudan University\\
Shanghai, 200433, China\\
zhangzz@fudan.edu.cn }
}

\markboth{IEEE International Conference on Data Engineering}
{Liu \MakeLowercase{\textit{et al.}}:\title{Resistance Eccentricity in Graphs: Distribution, Computation and Optimization}}

\IEEEtitleabstractindextext{%
\begin{abstract}
The advent of online social networks has facilitated fast and wide spread of information.
However, some users, especially members of minority groups, may be less likely to receive information spreading on the network, due to their disadvantaged network position. We study the optimization problem of adding new connections to a network to enhance fairness in information access among different demographic groups.

We provide a concrete formulation of this problem where information access is measured in terms of resistance distance, {offering a new perspective that emphasizes global network structure and multi-path connectivity.} The problem is shown to be NP-hard. We propose a simple greedy algorithm which turns out to output accurate solutions, but its run time is cubic, which makes it undesirable for large networks. As our main technical contribution, we reduce its time complexity to linear, leveraging several novel approximation techniques. In addition to our theoretical findings, we also conduct an extensive set of experiments using both real-world and synthetic datasets. We demonstrate that our linear-time algorithm can produce accurate solutions for networks with millions of nodes.

\end{abstract}
\begin{IEEEkeywords}
Social networks, resistance distance, efficiency, algorithmic fairness, edge recommendation, combinatorial optimization.
\end{IEEEkeywords}}

\maketitle

\IEEEdisplaynontitleabstractindextext

\IEEEpeerreviewmaketitle

\section{Introduction}
Social networks have become vital channels for communication, allowing users to share information quickly and connect with others globally. However, they can harbor biases and inequities for different demographic groups~\cite{beilinson2020clustering,bashardoust2022reducing}. It has been consistently observed that disadvantaged and minority groups, such as non-male workers in a male-dominated industry, have less access to life-altering information, such as job advertisements~\cite{fish2019gaps,tsang2019group,yaseen2016influence,speicher2018potential}, medical assistance~\cite{stoica2020seeding,bashardoust2022reducing,xu2022algorithmic}, and research/business ideas~\cite{jalali2020information}. Such structural inequity is particularly undesired since it can fuel a reinforcing cycle where well-connected groups have better opportunities for further improvement, perpetuating inequality~\cite{bashardoust2022reducing,jalali2020information}.

Consequently, there has been a growing interest in introducing countermeasures to amplify the information access for disadvantaged groups. Roughly speaking, the \textit{information access} is the likelihood of being exposed to a piece of information circulating in the network. This is usually measured by social capital and centrality indicators such as degree, closeness, and information  centrality~\cite{fish2019gaps,jalali2020information,Suh2022,LiZhZe24}. Then, the goal is to close the gap in information access level between advantaged and disadvantaged groups in a network. It is worth emphasizing that here ``fairness'' usually refers to ``equality''. Since this a new and emerging line of research, more complex notion of fairness still to be formulated and studied.

Existing attempts towards fair information access typically involve two steps: firstly, quantifying the degree of unfairness in information access~\cite{swift2022maximizing,anwar2021balanced,stoica2020seeding,tsang2019group,fish2019gaps,jalali2020information,tsioutsiouliklis2022link,tsioutsiouliklis2021fairness}; and secondly, devising proactive countermeasures~\cite{stoica2020seeding,swift2022maximizing}. One well-studied approach is to engineer the seed users for the information spread to achieve fairness~\cite{stoica2020seeding,tsang2019group}. However, these strategies have two fundamental shortcomings: they are temporary solutions (i.e., they need to be applied actively) and are not very practical since we cannot control the seed users in many scenarios. An alternative, that doesn't suffer from these shortcomings, is modifying the underlying network structure to resolve the inequality by enhancing the structural importance of the disadvantaged group, cf.~\cite{swift2022maximizing,bashardoust2022reducing,tsioutsiouliklis2022link}. However, one of the main challenges in this setting is the complexity bottlenecks, particularly given the massive size of many real-world networks.

Various graph parameters have been introduced to measure the centrality, social capital, and information access of a node (user) in a network, such as degree, closeness, PageRank~\cite{LU20161, GlDa15}. One measure which has gained substantial popularity is average resistance distance~\cite{li2019current, shan2018improve}, due to its ability to capture complex local and global network structures beyond other measures. Consequently, we use this measure to define the notion of fairness.

Some common graph operations, used to reach various objectives such as fairness, include edge/node deletion, addition, and rewiring~\cite{ren2018dismantling,shan2018improve,fairdrop21,li2021on,burst_2020,jalali2020information,swift2022maximizing,bashardoust2022reducing,tsioutsiouliklis2022link}. Among these, edge addition has absorbed most attention, cf.~\cite{santos2021link}, due to its practicality. While for example removing nodes (users) or edges (connections) is difficult in practice (e.g., this is in violation of freedom of expression), adding edges, or at least encouraging their formation, is more practical (e.g., via link recommendation systems).

We study the problem where the input is an $n$-node graph $G$ (representing a social network), a budget $k$, a disadvantaged group $T$ and advantaged group $S$, and we aim to minimize the gap between their information access (measured in terms of average resistance distance) by adding $k$ edges. We, in fact, aim to optimize fairness, while improving the overall information access as well. Please refer to Problem~\ref{pro:1} for the exact formulation.
Our contribution can be summarized as follows:

\begin{itemize}
    \item {We introduce a new perspective on fairness of information access by formulating it through effective resistance. We show that this problem is NP-hard and its objective function is not supermodular.}
    \item We propose a greedy approach which performs well, but its run time is cubic, making it impractical for large networks.
    \item Leveraging several novel approximation techniques, such as fast Laplacian solver and the high-dimensional convex hull approximation, we introduce a linear time algorithm.
    \item We conduct experiments on a large set of datasets. The experiments demonstrate that our proposed linear time algorithm not only produces accurate solutions, but also can manage networks with millions of nodes.
\end{itemize}
\section{Related Work}
\label{related-work} 
In this section, we briefly review the literature related to our work.

\noindent\textbf{Information Dissemination.} Social networks play a key role in spreading information, sharing news, and promoting lesser-known businesses. The efficiency of this process is influenced by the network's topology, which can be characterized by structural properties like spectral radius, algebraic connectivity, and mixing time~\cite{levin2017markov,freitas2022graph,LiZhZe24}. One important measure is the Kirchhoff index (the total pairwise resistance distance, see Definition~\ref{def:kir})~\cite{LiZh18}, which reflects network connectivity and has wide applications. A lower Kirchhoff index suggests better connectivity and is related to other metrics like commute and cover time~\cite{ChRaRu96}.

A wide spectrum of research studies has focused on enhancing the overall information access by reducing the Kirchhoff index. For instance, in~\cite{ghosh2008minimizing}, the authors investigated the minimization of Kirchhoff index by assigning edge weights, while in~\cite{wang2014improving,ZhAhZh25}, the authors examined minimizing the graph resistance distance by adding edges. However, these works optimized a single metric of global connectivity while neglected the \textit{fairness} aspect, which refers to whether the intervention would improve the information access of various groups in a fair manner. In this paper, we develop fast algorithms that simultaneously optimize the Kirchhoff index and fairness in information access.

\noindent\textbf{Algorithmic Fairness.}
Algorithmic fairness has been an active research topic in social network analysis, with systematic developments on the input, algorithm, and output~\cite{pitoura2021fairness,friedler2019comparative}. Despite the prevalence of graph-structured data in many areas~\cite{newman2018networks}, existing works usually overlooked them, focusing on independent and identically distributed data such as spatial or text data~\cite{feldman2015certifying,kamiran2012data,asudeh2020fairly,shetiya2022fairness}. In recent years, novel methods have emerged to extend traditional fairness notation and algorithms to graph data, including group-based fairness for centrality measures~\cite{tsioutsiouliklis2021fairness,tsioutsiouliklis2022link}, embeddings~\cite{dai2021say}, influence maximization~\cite{tsang2019group,swift2022maximizing,anwar2021balanced}, and clustering~\cite{kleindessner2019guarantees}, where groups are often divided by attributes such as gender, religion or race~\cite{tsioutsiouliklis2022link}. There are also individual fairness approaches based on the premise that similar nodes should be treated similarly~\cite{kang2020inform, dong2021individual}. Our work focuses on the group-based fairness of the graph algorithm input, namely the graph structure itself.

Furthermore, during recent years, fair access to information has emerged as a pressing concern. Information access, which is closely related to the notion of social capital~\cite{coleman1988social,fish2019gaps}, pertains to the acquisition and utilization of information from social networks and was initially introduced in~\cite{fish2019gaps}. Prior research on fair information access mainly relied on the influence maximization framework, which aims to maximize the influence of seed nodes while also satisfying certain fairness constraints through fair seeding~\cite{tsang2019group,stoica2020seeding,fish2019gaps} or edge recommendation~\cite{swift2022maximizing}. Another line of research has considered the network structure solely, such as decreasing information unfairness by drawing help of the access matrix~\cite{jalali2020information} or access signatures~\cite{bashardoust2022reducing}. However, their methods can not handle large networks due to the lack of fast algorithms. Our work differs from most of these works in two ways: We not only provide a novel and intuitive formulation of the problem building on the popular resistance distance, but also devise an efficient and accurate algorithm which can handle networks with several million nodes.
\section{Preliminaries}\label{sec:prelimi}
\subsection{Notations}
Unless otherwise specified, we denote scalars in $\mathbb{R}$ by normal lowercase letters like $a, b, c$, sets by normal uppercase letters like $A, B, C$, vectors by bold lowercase letters like $\aa, \bb, \cc$, and matrices by bold uppercase letters like $\AA, \BB, \CC$. Let $\mathbf{1}$ denote the vector of appropriate dimensions with all entries being ones. We use $\AA[i,:]$ and $\AA[:,j]$ to denote, respectively, the $i$-th row and the $j$-th column of matrix $\AA$. We write $\AA_{ij}$ to denote the entry at row $i$ and column $j$ of $\AA$ and $\aa_i$ to denote the $i$-th element of vector $\aa$. Let $\aa^{\top}$ and $\AA^{\top}$ denote, respectively, the transpose of vector $\aa$ and matrix $\AA$. Define $\trace{\AA}$ to be the trace of matrix $\AA$. An $n \times n$ matrix $\AA$ is positive semi-definite if $\xx^{\top} \AA \xx \geq 0$ holds for all $\xx \in \mathbb{R}^{n}$. We write $\ee_i$ to denote the $i$-th standard basis vector. Let $\EE^{S}$ represent a square diagonal matrix with the elements in $S$ being one and zero everywhere else. For an arbitrary square matrix $\AA$, we write $\AA_{S}$ to denote the sum of the diagonal elements of $\AA$ with respect to set $S$, i.e.,  $\AA_S = \sum_{i \in S} \AA_{ii}.$

\begin{definition}\label{def:eps-appr-mat}
Given two scalars $a,b > 0$, $b$ is called an $\eps$-approximation (abbr. $b \stackrel{\epsilon}{\approx} a$) of $a$ if $(1-\eps)\, a \leq b \leq (1+\eps)\, a$.
\end{definition}
The following facts are basic properties of $\eps$-approximation.
\begin{fact}\label{fac:1}
For non-negative scalars $a, b, c, d \geq 0$, if $a \stackrel{\epsilon_1}{\approx} b$, $c \stackrel{\epsilon_1}{\approx} d$, and $b \stackrel{\epsilon_2}{\approx} c$, then $a+c \stackrel{\epsilon_1}{\approx} b+c$, $a+c \stackrel{\epsilon_1}{\approx} b+d$, $a \stackrel{\epsilon_1+\epsilon_2}{\approx} c$, and $a c \stackrel{\epsilon_1}{\approx} b c$.
\end{fact}

Let $X$ be a finite set, and $2^X$ be the set of all subsets of $X$.
A set function $f: 2^X \to \mathbb{R}$ is called \textit{monotone} decreasing if for any subsets $B \subset C \subset X$, $f(B) > f(C)$ holds. For any subsets $B \subset C \subset X$ and any element $e\in X\setminus C$, we say the function $f$ is \textit{supermodular} if it satisfies $f(B) - f(B\cup \{e\}) \geq  f(C) - f(C\cup\{e\})$. A typical problem in supermodular optimization is the \textit{cardinality-constrained set function} problem, and a greedy-based algorithm has emerged as a popular option for tackling it with a guaranteed approximation ratio of $(1-1/e)$~\cite{nemhauser1978analysis}.

\subsection{Graphs and Relevant Matrices}
Consider a connected undirected graph $\calG = (V,E)$ where $V$ is the set of nodes and $E \subseteq V \times V$ is the set of edges. Let $n = |V|$ and $m = |E|$ denote the number of nodes and the number of edges, respectively. Let $\Bar{E}=(V\times V)\backslash E$ be the set of nonexistent edges in $\calG$ and $\bar{m}=\sizeof{\bar{E}}$ denote its size. The Laplacian matrix of $\calG$ is the symmetric matrix $\LL = \DD - \AA$, where $\AA$ is the adjacency matrix whose entry $\AA_{ij}=1$ if node $i$ and node $j$ are adjacent, and $\AA_{ij}=0$ otherwise, and $\DD$ is the degree diagonal matrix $\DD=\text{diag}(\dd_1,\cdots,\dd_n)$ where $\dd_i$ is the degree of node $i$. We fix an arbitrary orientation for all edges in $\calG$, then we can define the signed edge-node incidence matrix $\BB_{m\times n}$ of graph $\calG$, whose entries are defined as follows: $\BB_{e,u}= 1$ if node $u$ is the head of edge $e$, $\BB_{e,u}= -1$ if $u$ is
the tail of $e$, and $\BB_{e,u}= 0$ otherwise. For an oriented edge $e=(u,v)\in E$, we define $\bb_e = \bb_{uv} = \ee_{u}-\ee_{v}$, where $u$ and $v$ are head and tail of $e$, respectively. $\LL$ can be rewritten as $\LL = \sum\nolimits_{e\in E}\bb_e \bb_e^\top$. 
The matrix $\LL$ is singular and positive semi-definite and its Moore-Penrose pseudoinverse is $\LL^\dag = \kh{\LL +\frac{1}{n}\JJ}^{-1}-\frac{1}{n}\JJ$, where $\JJ$ is the matrix of appropriate dimensions with all entries being ones.

\subsection{Resistance Distance}
For an undirected connected network $\calG$, we could construct a corresponding electrical network, with each edge being replaced by a resistance of one ohm. Next, we introduce three quantities of resistance distance based on this electrical network.

\begin{definition}\label{def:resis_dis}
(Pairwise Resistance Distance~\cite{KlRa93}). For a graph $\calG = (V,E)$, {let $u,v\in V$ be two distinct nodes and $\bb_{uv}=\ee_u-\ee_v$}, the pairwise
resistance distance between $u$ and $v$ is defined as 
$\mathcal{R}_{uv}=\bb_{uv}^{\top} \LL^{\dagger} \bb_{uv}.$
\end{definition}

\begin{definition}\label{def:resis_dis2}
(Node Resistance Distance~\cite{BOZZOresistance2013}). For graph $\calG = (V,E)$, the resistance distance $\mathcal{R}_v$ of node $v$ is defined as the pairwise resistance distances between node $v$ and all nodes in $V$, i.e., 
$\mathcal{R}_v = \sum_{u \in V} \mathcal{R}_{uv} = n \LL^\dag_{vv} + \ldtrace{\LL^\dag}.$
\end{definition}

\begin{definition}
\label{def:kir}
(Kirchhoff index~\cite{LiZh18}). The Kirchhoff index $\calK(\calG)$ of a graph $\calG$ is defined as the sum of resistance distances over all node pairs, which is a measure of the overall connectedness of a network; namely,
$
\calK(\calG)=\sum_{\substack{u, v \in V \\u<v}} \mathcal{R}_{uv}=n \ldtrace{\LL^{\dag}}.  
$
\end{definition}

{\noindent\textbf{Intuition.} Lower resistance distance between nodes implies easier communication and exchange of information. This measure has particularly gained significant popularity since it is capable of encoding complex local and global structure which impact information spread. Two complementary views make this concrete. 
\textit{Electrical view:} injecting one unit of current at $u$ and removing it at $v$ yields a voltage drop equal to $\mathcal{R}_{uv}$. 
\textit{Random-walk view:} the commute time satisfies $\mathrm{CommuteTime}(u,v)=2m\,\mathcal{R}_{uv}$, so smaller resistance means easier mutual reachability and faster diffusion. Please refer to~\cite{ChRaRu89,stephenson1989rethinking,brandesCentralityMeasuresBased2005} for more details.}
\subsection{Convex Hull}\label{sec:convex}
The convex hull is the smallest convex polygon enclosing a set of points, which plays a significant role in computational geometry~\cite{ChSk04}, computer vision~\cite{yang2013graph}, machine learning~\cite{fawcett2007pav} and so on.

\begin{definition}
(Convex Hull~\cite{Ro70}) 
Given a set of $n$ points $P=\{\vvv_1,\vvv_2,\cdots,\vvv_n\}$ in $\mathbb{R}^d$, its convex hull $C(P)$ is the (unique) minimal convex polytope containing $P$, providing the tightest convex boundary encompassing these points. Let $\bar{P}$ denote the node set of $C(P)$.
\end{definition}

One important characteristic of a convex hull is the \textit{diameter} property which asserts that the maximum distance between any two points within $P$, i.e., the diameter of $P$, is equal to the maximum distance between two points in $\Bar{P}$. In other words, let $D_P= \max_{i,j\in P}\|\vvv_i-\vvv_j\|$, and $D_{\bar{P}}= \max_{i,j\in \bar{P}}\|\vvv_i-\vvv_j\|$ denote the diameter of $P$ and $\bar{P}$ respectively, then $D_P=D_{\bar{P}}$. 

\section{PROBLEM FORMULATION}
\label{problem-form}

\subsection{Definitions}
In social networks, different groups are often found to have unequal access to information due to a range of factors such as socio-economic status, demographics, or other factors that shape their position within a network. This issue has been explained and highlighted in prior work~\cite{bashardoust2022reducing,jalali2020information}. As the first step to address this issue, we formulate the information access of a group in Definition~\ref{def:1}.

\begin{definition}\label{def:1}
    ($I$, Group Information Access). For any { non-empty} group $S\subseteq V$ in $\calG$, we measure its capability of accessing information from the network using the average resistance distance, which is the sum of resistance distances of the nodes in $S$ divided by its size, that is 
    \begin{align}
     \notag I_{S} = \frac{\sum_{v\in S}\mathcal{R}_v}{\abs S} = \frac{n}{\abs S}\LL^\dag_{S}+ \ldtrace{\LL^\dag }.
    \end{align}
\end{definition}
The parameter $I_{S}$ measures the information access capability of a group $S$. The smaller the value of $I_{S}$, the better access the group has to the information disseminated in the network. {The rationale is that resistance distance reflects the ``effective proximity'' of a node to the rest of the graph by aggregating contributions from all paths. Averaging these scores over a group yields a stable and interpretable notion of how easily the group can exchange information with the network.}  We call a group advantaged if it has a lower $I$ and disadvantaged if it has a higher $I$. To simplify our analysis, we focus on the two-category case in this paper, where nodes are categorized as either advantaged or disadvantaged. Nevertheless, our algorithms can easily be extended to cover cases where nodes belong to more than two categories. It is noteworthy that in the real-world social networks, categorization into two groups based on factors such as gender is common.
\begin{definition}\label{def:2}
    ($U$, Information Access Unfairness). For an advantaged group $S$ and a disadvantaged group $T$ in $\calG$, the extent of unfairness in information access is defined as
    \begin{align}
        \notag U_{\calG} = I_{T} - I_{S}. 
    \end{align}
\end{definition}
The intuition behind the definition of $U$ is to measure to what extent the network's structure enables equitable access to information among different groups. A value closer to 0 indicates a fairer network, whereas a larger value indicates more significant disparity between groups. Specifically, when $U_{\calG}=0$, the network satisfies { group equality~\cite{tsioutsiouliklis2022link}}.

\begin{definition}\label{def:3}
    ($R$, Graph Resistance).
We define the resistance of a graph $\calG=(V,E)$ to be
    \begin{align}
     \notag R_{\calG}=\frac{I_{V}}{2} = \frac{\sum_{v\in V}\mathcal{R}_v}{2n} = \ldtrace{\LL^\dag }.
    \end{align}
\end{definition}
The graph resistance $R$ is closely related to the Kirchhoff index in Definition~\ref{def:kir}.

\noindent\textbf{Remark.} As we mentioned, graph resistance and efficient information dissemination have inverse relation. Thus, we interchangeably use the terminologies of \textit{minimizing graph resistance} and \textit{maximizing information access efficiency}.

\subsection{Fair Information Access: Problem Statement}

We write $I_{S}(\add)$ and $I_{T}(\add)$ to denote the resulting average resistance distance of group $S$ and $T$ respectively after adding the edges in $\add$, and $R_{\calG}(\add)$ to denote the resulting overall graph resistance.
We study the following objective function:
\begin{align}\label{eq:obj}
    F(\add)=(1-\lambda)R_{\calG}(\add)+\lambda \big(I_{S}(\add)^2+I_{T}(\add)^2\big)
\end{align}
where the hyperparameter $\lambda$ controls the trade-off between the minimization of graph resistance and unfairness. 

{\noindent\textbf{Rationale of the objective function.} Eq.~\eqref{eq:obj} {couples} a network–efficiency term with a group–aware term. 
Minimizing the effective resistance $R_{\calG}$ alone is group–agnostic and may perpetuate existing advantages, whereas selecting edges by group access alone can myopically inflate local access without meaningfully shortening network-wide paths, yielding diminishing \textit{global returns}. 
We therefore adopt a standard bi–criteria scalarization that (i) reduces global impedance ($R_{\calG}$) and (ii) steers access in a group–aware manner.

For the fairness part, we optimize the {sum of squared} group access scores, $I_S^2+I_T^2$, rather than the raw disparity $U:=I_T-I_S$ in Def.~IV.2. 
Directly optimizing $U$ is ill–behaved (it can be negative and is non–monotone), while $|U|$ is non–smooth and hinders large–scale optimization. 
In contrast, $I_S^2+I_T^2$ is nonnegative, smooth, and—critically—decomposes cleanly into edge–wise contributions, enabling scalable marginal estimation. 
This surrogate is also linked to disparity  (Cauchy–Schwarz):
for $I_S,I_T\ge 0$,
\(
|U|=|I_T-I_S|\ \le\ \sqrt{2}\,\sqrt{I_S^2+I_T^2},
\)
so lowering the squared sum tightens a direct bound on $|U|$. 
Intuitively, the first–order change 
$\mathrm{d}(I_S^2+I_T^2)=2I_S\,\mathrm{d}I_S+2I_T\,\mathrm{d}I_T$
weights improvements by the current group levels, which biases selections toward reducing disparity; in our experiments this aligns with edges that benefit the disadvantaged group more.
The above design is consistent with common practice in fairness optimization~\cite{anwar2021balanced,ali2019fairness}.
(When a third group $O$ exists, we extend the second term by adding $I_O^2$.)
}

For a connected undirected graph $\calG=(V,E)$, if we add a set $\add$ of nonexistent edges to $\calG$ forming a new graph $\calG(\add)=(V, E\cup\add)$, the objective function $F$ will decrease. We next prove this property. 
\begin{lemma}\label{lem:mono}
    Let $\calG=(V,E)$ be a connected graph and $e\notin E$ be a potential edge to be added. Let $S, T\subseteq V$ be two demographic groups. Define $\Delta(e)= F(\emptyset)-F(\{e\})$ as the decrease of $F$ upon the addition of edge $e$. Then $\Delta(e)>0$.
\end{lemma}
\begin{proof} 
Note that $F$ is defined as a non-negative linear combination of resistance distance. By Rayleigh’s monotonicity law~\cite{doyle_snell_1984}, adding edges can only decrease resistance distance. Therefore,  we conclude that $\Delta(e)>0$.
\end{proof}

Lemma~\ref{lem:mono} indicates that the addition of any non-existing edge will lead to a decrease of the objective function $F$. Then, a question naturally arises: How can we find the optimal edge set $\add$ with $k$ edges such that $F$ is minimized? More precisely, the fair information access maximization problem can be stated as follows.

\begin{problem}\label{pro:1}(\uline{F}air \uline{I}nformation \uline{A}ccess \uline{M}aximization, \MFI). Given a connected undirected graph $\calG=(V,E)$, two demographic groups $S,T\subseteq V$, a fixed integer $k$, we aim to select the edge set $\add\subseteq \bar{E}$ with $\sizeof{\add}=k$, and add these chosen edges to graph $\calG$ forming a new graph $\calG(\add)=(V, E\cup \add)$, so that $F$ is minimized. In the other words:
\begin{equation}\label{pro1}
\add = \argmin_{H \subseteq \bar{E}, |H|=k} F(H).
\end{equation}
\end{problem}

\subsection{Hardness and Non-Supermodularity}
In this section, we argue that Problem~\ref{pro:1} is NP-hard. In~\cite{kooij2023minimizing}, the authors proved that the problem of minimizing the Kirchhoff index of a graph by adding $k$ edges is NP-hard. This proof relies on a reduction from the $3$-colorability problem~\cite{garey1974some}. Our problem is at least as hard as theirs.
Note that by setting $\lambda=0$ in Eq.~\eqref{eq:obj}, the quantity that we optimize differs from the Kirchhoff index only by a multiplicative factor, according to Definitions~\ref{def:kir} and~\ref{def:3}. Therefore, we can conclude that the \MFI problem is NP-hard. This implies that we cannot hope for a polynomial time algorithm. Furthermore, the objective function in Problem~\ref{pro:1} is not supermodular, as stated below.
\begin{figure}[t]
    \centering
    \includegraphics[width=0.7\columnwidth]{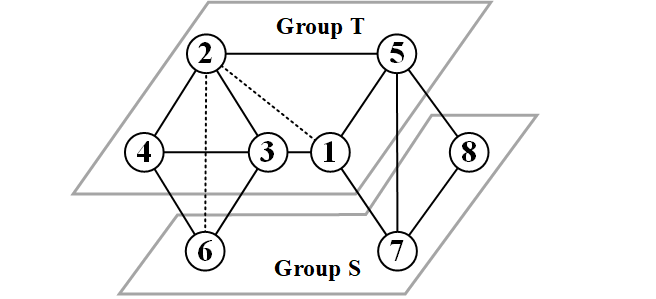}
    \caption{A $8$-node toy network to show the non-supermodularity of the objective function.}
    \label{fig:nonsuper}
\end{figure}
\begin{theorem}\label{the:super}
    Function $F(\cdot)$ is not supermodular.
\end{theorem}
\begin{proof}
    To illustrate the non-supermodularity, consider the network with $8$ nodes in Fig.~\ref*{fig:nonsuper}, where the solid lines represent the existing edges and the dashed lines represent the edges to be added. Define $B=\emptyset$, $C=\{(1,2)\}$ and $e=(2,6)$. Then, we have $F(B)= 64.58, F(B\cup\{e\}) = 56.45, F(C) = 56.88$, and $F(C\cup\{e\}) = 48.70$. Thus,
    $
    F(B) - F(B\cup \{e\}) =  8.13 \leq  F(C) - F(C\cup\{e\}) = 8.18.
    $
    This concludes the proof.
\end{proof}

\section{Simple Greedy Algorithm}
\label{simple-greedy}
The \MFI problem is combinatorial, and its optimal solution can be found using a brute-force approach. This involves computing \(F(\add)\) for each set \(\add\) of the \({\bar{m}\choose k}\) candidates on the augmented graph \(\calG(\add)\), requiring \(\Theta(n^3)\) time to invert the Laplacian matrix. The method outputs the subset \(\add^{*}\) of \(k\) edges that maximally decreases \(F\). However, this is computationally infeasible due to its time complexity of \(\Theta\left({\bar{m}\choose k}n^3\right)\), which grows exponentially with \(k\).

{To address the exponential complexity, we observe that the objective $F$ decreases most when we add the non-edge whose \textit{one-step} improvement
$\Delta(e)$ is largest. This motivates a discrete analogue of steepest descent: directly optimize the immediate decrease in $F$. 
Accordingly, we propose a simple greedy algorithm, \Greedy, outlined in Algorithm~\ref{alg:greedy}.} (It is worth emphasizing that, despite its name, this algorithm does not compute an exact solution, but it's called \Greedy since we use further approximation to speed up this algorithm in next sections.) Initially, we compute the inversion of the Laplacian matrix \(\LL\) in \(\mathcal{O}(n^3)\) time and set the edge set \(\add\) to empty. If \(\LL^\dag\) is known, adding an edge can be treated as a rank-1 update, which can be done in \(\mathcal{O}(n^2)\) time using the Sherman-Morrison formula~\cite{Me73}. The algorithm proceeds in \(k\) rounds, each involving two operations: computing \(\Delta(e)\) in \(\mathcal{O}(\Bar{m}n)\) time and updating \(\LL^{\dag}\) in \(\mathcal{O}(n^2)\) time. In each iteration, edge \(e_i\) is chosen to minimize the objective function \(F\), terminating after \(k\) edges are selected. The overall running time is \(\mathcal{O}(n^3 + k\Bar{m}n)\).

\normalem
\begin{algorithm}[h]
	\caption{$\Greedy(\calG, k)$}\label{alg:greedy}
	\Input{
		A graph $\calG=(V,E)$; an integer $0<k \leq \bar{m}$
	}
	\Output{
		A subset of $\add \subseteq \bar{E}$ and $|\add| = k$
	}
	Compute $\LL^{\dag}$;
	$\add \gets \emptyset$ \;
	\For{$i = 1$ to $k$}{
		Select $e_i$ s.t. $e_i \leftarrow \argmax_{e \in \Bar{E}\backslash \add} \Delta(e)$\;
		Update solution $\add \gets \add \cup \{e_i\}$\;
		Update $\LL^\dag \gets \LL^\dag- \frac{\LL^\dag \bb_{e_i} \bb_{e_i}^\top \LL^\dag}{1 +\bb_{e_i}^\top \LL^\dag \bb_{e_i}}$\;
        Update the graph $\calG\gets (V, E\cup\{e_i\})$
	}
	
    \Return $\add$ \;
\end{algorithm}

\section{FAST GREEDY ALGORITHM}
\label{fast-greedy}
While \Greedy is faster than the naïve brute-force approach, its time complexity remains suboptimal due to computation challenges such as: 1) computing the pseudo-inverse of \(\LL\); 2) traversing nearly quadratic unconnected node pairs in real-world networks; and 3) updating \(\LL^{\dag}\) after each edge selection.
In this section, we propose a novel, faster algorithm for solving \MFI, utilizing pivotal elements that we will introduce and analyze in the following subsections.

\subsection{High-level Ideas}
Our approach works with a two-stage pipeline that turns ``best marginal'' into a fast distance search while keeping explicit accuracy guarantees. 

\noindent\textbf{Transforming decrement search to distance search (Section~\ref{sec:gradient}).}
While \Greedy\ selects the edge with the maximum marginal decrease, prior work (e.g., \cite{kang2019n2n,siami2017centrality,YiSh2018,wong2015efficiency}) has used gradients to assess an edge’s influence on an objective. Motivated by this, we use the first-order derivative of $F$ with respect to the weight of $e$ as a surrogate for the true marginal decrease $\Delta(e)$. Although this surrogate may not perfectly preserve the ranking of all edges, it enables a much faster algorithm and, empirically, strongly correlates with the true marginals.

We view $\calG$ as a weighted complete graph and treat $F$ as a multivariate function. For each non-edge $e\in (V\times V)\setminus E$, we compute the derivative to obtain the surrogate $\bar\Delta(e)$ (Lemma~\ref{lem:delta_bar}). This quantity admits a geometric interpretation: it equals the squared Euclidean distance between the high-dimensional node embeddings with coordinates $\cc_i$ defined in Eq.~\eqref{eq:zuobiao}. Consequently, Problem~\ref{pro:1} reduces to selecting the farthest pair in the point set $P=\{\cc_i\}_{i\in V}$ (Problem~\ref{pro:2}).
\normalem
\begin{algorithm}[t]
	\caption{$\Fast(\calG, \epsilon, k)$}\label{alg:fast}
	\Input{
		A connected graph $\calG$; an error parameter $\epsilon\in(0,1)$; an integer $0<k\leq \bar{m}$
	}
	\Output{
        A subset of $\add \subseteq \bar{E}$ and $|\add| = k$
	}
        Initialize solution $\add=\emptyset$\;
        \For{$b=1$ to $k$}{
        $(i, j)\gets \Farthest(\calG, \epsilon)$\;
        Update solution $\add\gets \add\cup \{(i, j)\}$\;
        Update the graph $\calG\gets\calG(V,E\cup\{(i, j)\})$\;
        }
    \Return $\add$\;
\end{algorithm}

\noindent\textbf{Fast convex hull approximation (Section~\ref{sec:convex2}).}
Because farthest pairs lie on the convex hull (Sec.~\ref{sec:convex}), we can restrict comparisons to hull vertices instead of scanning all $O(n^2)$ pairs. Computing the exact hull in high dimensions is prohibitive, so we use an $c$-node \emph{approximate} convex-hull routine with accuracy guarantees, and further sketch the coordinates to a low dimension to preserve pairwise distances up to a small multiplicative error. We evaluate the required coordinates via a fast Laplacian solver. Putting these pieces together, our \Farthest algorithm (Algorithm~\ref{alg:furthest}) returns a pair whose distance is within a small multiplicative factor of the true diameter of $P$.

\begin{figure}[t]
    \centering
    \includegraphics[width=0.8\columnwidth]{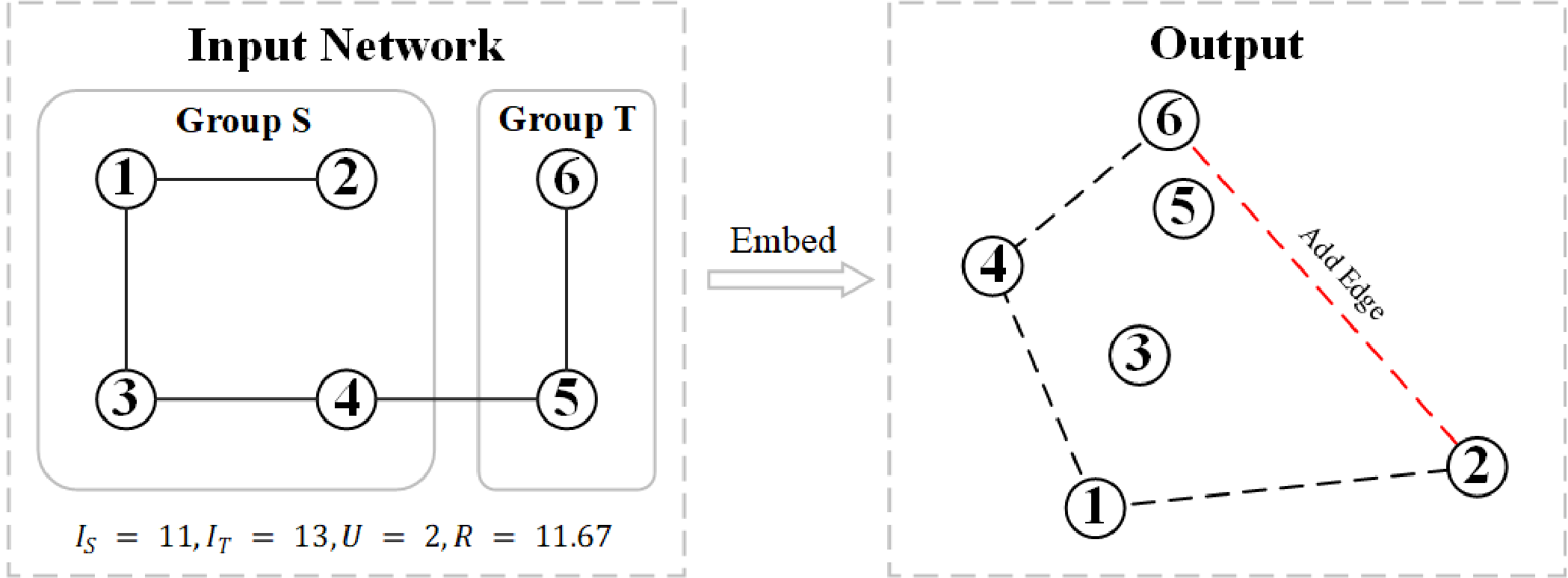}
    \vspace{-0.2cm}
    \caption{Running example on a toy graph. Embeddings are visualized in 2D for clarity. }
    \label{fig:pipeline}
\end{figure}

\noindent\textbf{The \Fast algorithm.} By applying \Farthest to iteratively find the best local unconnected node pair whose connection could maximally decrease the objective function, we propose a fast
greedy algorithm called \Fast (Algorithm~\ref{alg:fast}).

To make the pipeline concrete, we show a toy instance in Fig.~\ref{fig:pipeline} where the node embeddings are visualized in two dimensions.
From the surrogate view, the best edge corresponds to the farthest pair in the embedded set.
As illustrated, we first compute the embeddings, then restrict attention to the (approximate) convex hull (dashed), and finally select the farthest pair—nodes $(2,6)$—and add that edge to the graph.
This single picture mirrors the main steps of our method.

\subsection{Gradient of Marginal Decrease} 
\label{sec:gradient}
To compute the gradient, we consider the graph \(\calG\) as a complete graph, where nonexistent edges have a weight of 0 and existing edges have a weight of 1. Let \(\tilde{m} = \tbinom{n}{2}\), \(\ww_0 \in [0,1]^{\tilde{m}}\) be the original graph's edge weights, and \(\ww \in [0,1]^{\tilde{m}}\) be the augmented graph's edge weights. We can rewrite Problem~\ref{pro:1} by taking \(F\) as a multivariate function of vector \(\ww\):
\begin{align}
    \label{pro1_vec}
\add = \argmin_{\ww\in [0,1]^{\tilde{m}},\ww-\ww_0\in\{0,1\}^{\tilde{m}}, \mathbf{1}^{\top}(\ww-\ww_0)=k} F(\ww).
\end{align}

{\begin{definition}[Gradient Surrogate]
\label{lem:delta_bar}
    Let $\alpha=1-\lambda+2\lambda\big(2\ldtrace{\LL^{\dag}}+\tfrac{n}{\abs S}\LL^{\dag}_{S}+\tfrac{n}{\abs T}\LL^{\dag}_{T}\big)$, $\beta=2\lambda \tfrac{n}{\abs S}\big(\ldtrace{\LL^{\dag}}+$ $\tfrac{n}{\abs S}\LL^{\dag}_{S}\big)$, and $\gamma=2\lambda \tfrac{n}{\abs T}\big(\ldtrace{\LL^{\dag}}+\tfrac{n}{\abs T}\LL^{\dag}_{T}\big)$. Then we can derive an surrogate of $\Delta(e)$, denoted by $\Bar{\Delta}(e)$, with
\begin{align}
     \Bar{\Delta}(e)=|\tfrac{\partial{F}}{\partial{\ww_i}}|=\alpha\|\LL^{\dag}\bb_e\|^2+\beta\|\EE^{S}\LL^{\dag}\bb_e\|^2+\gamma\|\EE^{T}\LL^{\dag}\bb_e\|^2\notag.
\end{align}
\end{definition}}

We could find that $\bar{\Delta}(e)$ is a square of Euclidean distance between nodes $i,j$ in a $2n$-dimensional space where the coordinate $\cc_i$ of a node $i\in V$ is as follows:
\begin{align}
\label{eq:zuobiao}
    \notag \cc_i=&\Big(\sqrt{\alpha}\ee_i^{\top}\LL^{\dag}, \underbrace{\sqrt{\beta}\ee_{v_{1}}^{\top}\LL^{\dag}\ee_i, \cdots, \sqrt{\beta}\ee_{v_{\abs{S}}}^{\top}\LL^{\dag}\ee_i}_{v_{j}\in S},\\
    &\underbrace{\sqrt{\gamma}\ee_{u_{1}}^{\top}\LL^{\dag}\ee_i, \cdots, \sqrt{\gamma}\ee_{u_{\abs{T}}}^{\top}\LL^{\dag}\ee_i}_{u_{j}\in T}\Big)^{\top}.
\end{align}

Finding an edge whose addition maximally decreases the objective function \( F \) simplifies to identifying two nodes with maximum Euclidean distance based on their coordinates. This can be formally defined as follows.

\begin{problem}\label{pro:2}
Given a graph \(\calG=(V,E)\) with \(n\) nodes and a set of \(n\) points \(P=\{\cc_1,\cc_2,\cdots,\cc_n\}\) in \(\mathbb{R}^{2n}\), where \(\cc_i\) is the coordinate of node \(i\), find two points in \(P\) with the maximum Euclidean distance.
\end{problem}

\noindent\textbf{$\Gradient$ Algorithm.} To address this, we propose the \(\Gradient\) algorithm, which substitutes the marginal decrease \(\Delta(e)\) in Line 4 of Algorithm~\ref{alg:greedy} with the distance square \(\bar{\Delta}(e)\). However, iterating over all unconnected node pairs to find the maximum distance is computationally inefficient for sparse networks, necessitating a more efficient method to identify the furthest node pair.

\input{farthest}
\input{analysis}
\begin{figure*}[ht!]
    \centering
    \includegraphics[width=\textwidth]{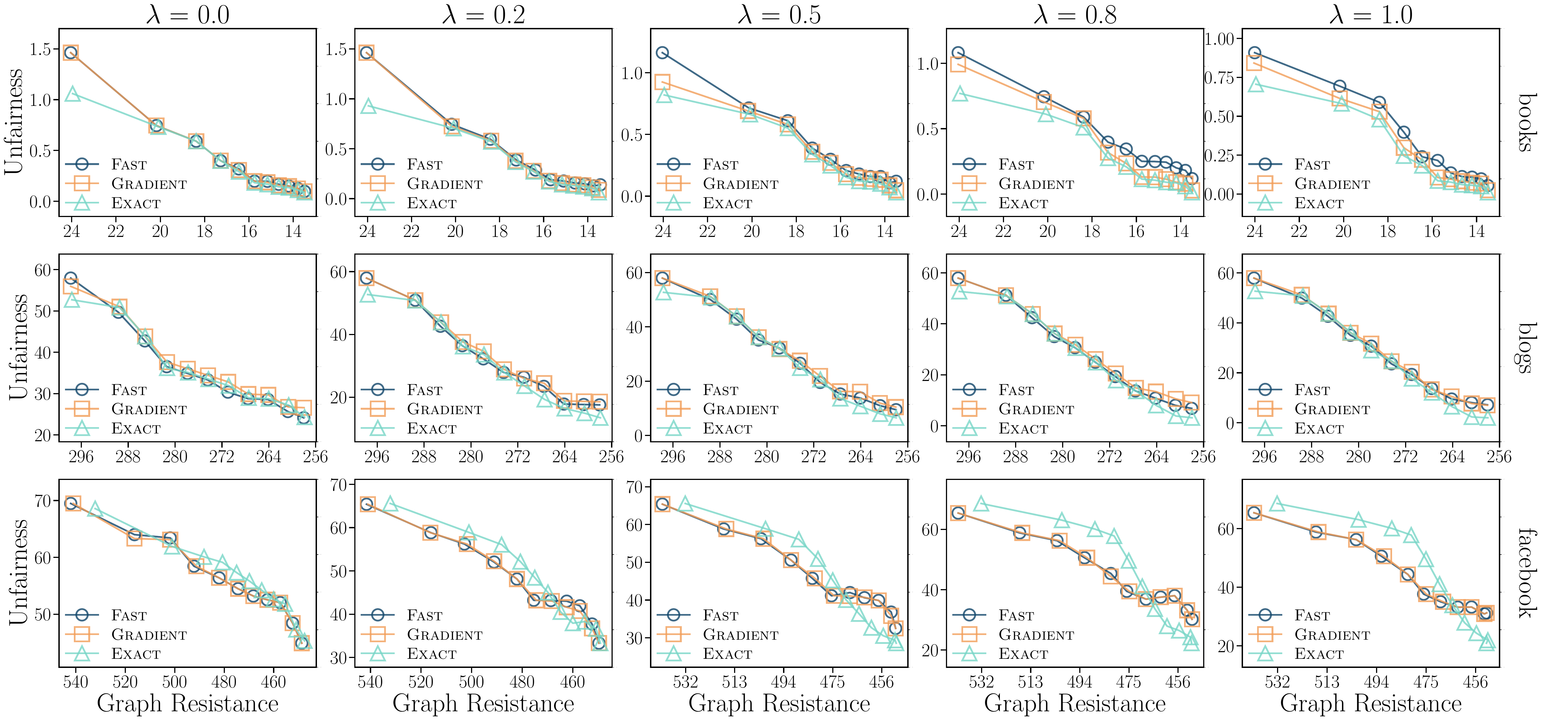}
    \vspace{-0.3cm}
    \caption{Graph resistance and unfairness as a function of the number of added edges for our algorithms. 
    }
    \vspace{-0.3cm}
    \label{fig:balance}
\end{figure*}

\section{Experiment}
\label{experiments}
We have conducted an extensive set of experiments on a large group of real-world and synthetic networks to evaluate the  effectiveness and efficiency of our algorithms, namely \Greedy (see Algorithm~\ref{alg:greedy}), \Gradient (see the end of Subsection~\ref{sec:gradient}), and \Fast (see Algorithm~\ref{alg:fast}). 
\begin{table}[t]
\begin{center}
\caption{Characteristics of the real-world networks. Parameters ratio (dis) and $U$ represent the population ratio of the disadvantaged group and the initial unfairness respectively.}
\label{tab:data}
\normalsize
\resizebox{1\columnwidth}{!}{
\fontsize{7}{8}\selectfont
\begin{tabular}{ccccc}
\Xhline{1\arrayrulewidth}
\raisebox{-0.5ex}{Network} & \raisebox{-0.5ex}{$n$} & \raisebox{-0.5ex}{$m$} & \raisebox{-0.5ex}{ratio (dis)} & \raisebox{-0.5ex}{$U$}\\[0.5ex]
\hline
\raisebox{-0.5ex}{books} & \raisebox{-0.5ex}{92} & \raisebox{-0.5ex}{748} & \raisebox{-0.5ex}{0.467} & \raisebox{-0.5ex}{3}\\[0.5ex]
\raisebox{-0.5ex}{blogs} & \raisebox{-0.5ex}{1,222} & \raisebox{-0.5ex}{16,717} &  \raisebox{-0.5ex}{0.485}& \raisebox{-0.5ex}{56} \\[0.5ex]
\raisebox{-0.5ex}{facebook} & \raisebox{-0.5ex}{4,039} & \raisebox{-0.5ex}{88,234} & \raisebox{-0.5ex}{0.388}& \raisebox{-0.5ex}{67} \\[0.5ex]
\raisebox{-0.5ex}{tmdb} & \raisebox{-0.5ex}{8,490} & \raisebox{-0.5ex}{207,036}& \raisebox{-0.5ex}{0.692}  & \raisebox{-0.5ex}{359}\\[0.5ex]
\raisebox{-0.5ex}{dblp-course} & \raisebox{-0.5ex}{13,015} & \raisebox{-0.5ex}{75,888}& \raisebox{-0.5ex}{0.172} & \raisebox{-0.5ex}{462} \\[0.5ex]
\raisebox{-0.5ex}{dblp-pub} & \raisebox{-0.5ex}{16,501} & \raisebox{-0.5ex}{133,226}& \raisebox{-0.5ex}{0.080}  & \raisebox{-0.5ex}{302}\\[0.5ex]
\raisebox{-0.5ex}{dblp-gender} & \raisebox{-0.5ex}{16,501} & \raisebox{-0.5ex}{66,613}& \raisebox{-0.5ex}{0.257} & \raisebox{-0.5ex}{429} \\[0.5ex]
\raisebox{-0.5ex}{dblp-aminer} & \raisebox{-0.5ex}{423,469} & \raisebox{-0.5ex}{1,231,211}& \raisebox{-0.5ex}{0.185} & \raisebox{-0.5ex}{11,373} \\[0.5ex]
\raisebox{-0.5ex}{linkedin} & \raisebox{-0.5ex}{3,209,448} & \raisebox{-0.5ex}{6,508,228}& \raisebox{-0.5ex}{0.627} & \raisebox{-0.5ex}{41,979} \\[0.5ex]
\Xhline{1\arrayrulewidth}
\end{tabular}
}
\end{center}
\vspace{-10pt}
\end{table}
\subsection{Experimental Setup}

\noindent\textbf{Machine Configuration.} All algorithms are implemented in Julia and run on a Linux box with 4.2 GHz Intel i7-7700 CPU and 32GB of main memory. All experiments are executed using a single thread. 
\noindent\textbf{Baselines.} We compare our proposed algorithms against the following algorithms, which are based on classical centrality measures or methods for minimizing the Kirchhoff index~\cite{wang2014improving}.
\begin{itemize}[leftmargin=*]
\item \textsc{Random}: randomly choose $k$ unconnected node pairs from two groups $S$ and $T$.
\item \textsc{Optimum}: 
the optimal strategy has two variants, \textsc{Optimum-R} and \textsc{Optimum-U} for choosing $k$ non-adjacent node pairs whose addition maximally decreases the graph resistance and unfairness respectively using a brute-force searching method.
\item \textsc{DP}, \textsc{BP} (DS, BS): choose $k$ non-adjacent node pairs with the maximum product (summation) of degrees, betweennesses, for two nodes from distinct groups.
\item  \textsc{Fiedler}: choose $k$ non-adjacent node pairs with the maximum difference of the Fiedler value from distinct groups.
\item  \textsc{ER}: choose $k$ non-adjacent node pairs with the maximum resistance distance from distinct groups.
\end{itemize}

Note that the idea of the last six baselines comes from~\cite{wang2014improving} but with proper adjustment to our problem setup. We have evaluated similar strategies to \textsc{DP} and \textsc{DS} for closeness and PageRank centrality. However, the performance is very poor, and therefore we do not report their results in our plots.

\noindent\textbf{Performance Measures.}
We focus on three measures: 1) the unfairness between different groups (i.e., $U$), where a smaller absolute value is preferred; 2) the graph resistance $R$, where a lower value is desired; 3) the running time of the algorithms.

\noindent\textbf{Datasets.} To evaluate our proposed strategies, we conduct experiments on a corpus of real-world networks (see Table~\ref{tab:data}) sourced from SNAP~\cite{LeSo16} and Network Repository~\cite{RoAh15}, as well as synthetic networks generated by the Barabasi-Albert-homophily (BAh) model~\cite{lee2019homophily} which has been frequently used in graph fairness related works~\cite{anwar2021balanced,Wang2022,HomoGla}. Note that nodes in the real-world networks come with inherent attribute information, which determines the groups $S$ and $T$.

\noindent\textbf{Network Generation Model.}
\label{appendix:ba}
Different random graph models have been proposed to mimic real-world social networks, cf.~\cite{friedrich2018diameter}. Such models are usually tailored to possess fundamental properties consistently observed in real-world networks, such as small diameter and power-law degree distribution. (One advantage of such synthetic graph structures over the real-world networks is that they give us the freedom to generate graphs of various sizes.) We have adopted a variant of the well-known Barabasi-Albert model called the Barabasi-Albert-homophily (BAh) model~\cite{lee2019homophily}. This variant has been frequently used in graph fairness related works~\cite{anwar2021balanced,Wang2022,HomoGla}. It incorporates a binary group attribute for each node, distinguishing between majority and minority groups. The probability of a newly added node belonging to the minority group is denoted by $f_a$ (for some $f_a < 0.5$) and the probability of attachment $\Pi_{i j}$ between new and existing nodes depends on both degree (as in the original BA model) and a homophily parameter $h$, as follows
\begin{align}
\Pi_{i j} \propto \begin{cases}h \dd_j & if \quad Group(i)=Group(j) \\ (1-h) \dd_j & if \quad Group(i) \neq Group(j)\end{cases}
\end{align} The parameter $h$ controls the level of homophily in the network, with $h > 0.5$ generating homophilic networks and $h < 0.5$ generating heterophilic networks. Choosing $h = 0.5$ results in the standard BA model and $h = 1$ results in the network fragmenting into two components corresponding to the two groups. We set $f_a=0.3$ and $h=0.7$ in all our experiments. We use BAhx to represent the BAh network with x nodes.

\subsection{Trade-off Between Fairness and Efficiency}
We vary the hyperparameter $\lambda\in\{0.0,0.2,0.5,0.8,1.0\}$ (higher values give more weight to fairness) to assess its impact on the trade-off between information access efficiency and fairness while iteratively adding \(k=1,2,\ldots,50\) edges.
The outcomes are demonstrated in Fig.~\ref{fig:balance} where we plot a marker every 5 data points. The $x$-axis represents graph resistance $R$ (the opposite of information access efficiency) and the $y$-axis represents unfairness $U$. The markers move from left to right to indicate the gradual addition of edges.

We observe as \(\lambda\) increases, \(U\) consistently decreases for the same number of added edges, though this comes at a slight increase in \(R\). Notably, the reduction in \(U\) is significant, while the increase in \(R\) is minimal. Thus, interestingly, one can significantly improve fairness with little comprise on overall information access efficiency.

Let us particularly investigate the outcome of \Greedy at two extremes: \(\lambda=0.0\) (disregarding fairness) and \(\lambda=1.0\) (solely optimizing fairness). For the blogs dataset, while the final value of $R$ for the case of $\lambda=0.0$ is only slightly lower than that for $\lambda=1.0$, the final values of unfairness $U$ are significantly different with $U=24$ for the former case and $U=4$ for the latter. Similar behavior is observed for the facebook dataset. More specifically, the value of \(U\) drastically reduces from 45 to 20 while the value of \(R\) increases slightly from 447 to 451.

\subsection{Effectiveness}
We evaluate our algorithms against the optimal and random strategies on small BAh networks with 5, 10, 20, and 30 nodes (which allow the computation of optimal solutions due to their size). For each graph size, we add 4 edges and averaged results over 100 randomly generated instances. Results in Fig.~\ref{fig:opt} demonstrate that our greedy algorithms consistently produce solutions close to optimal, significantly outperforming the random scheme. 

\begin{figure}[t]
    \centering
    \includegraphics[width=0.9\columnwidth]{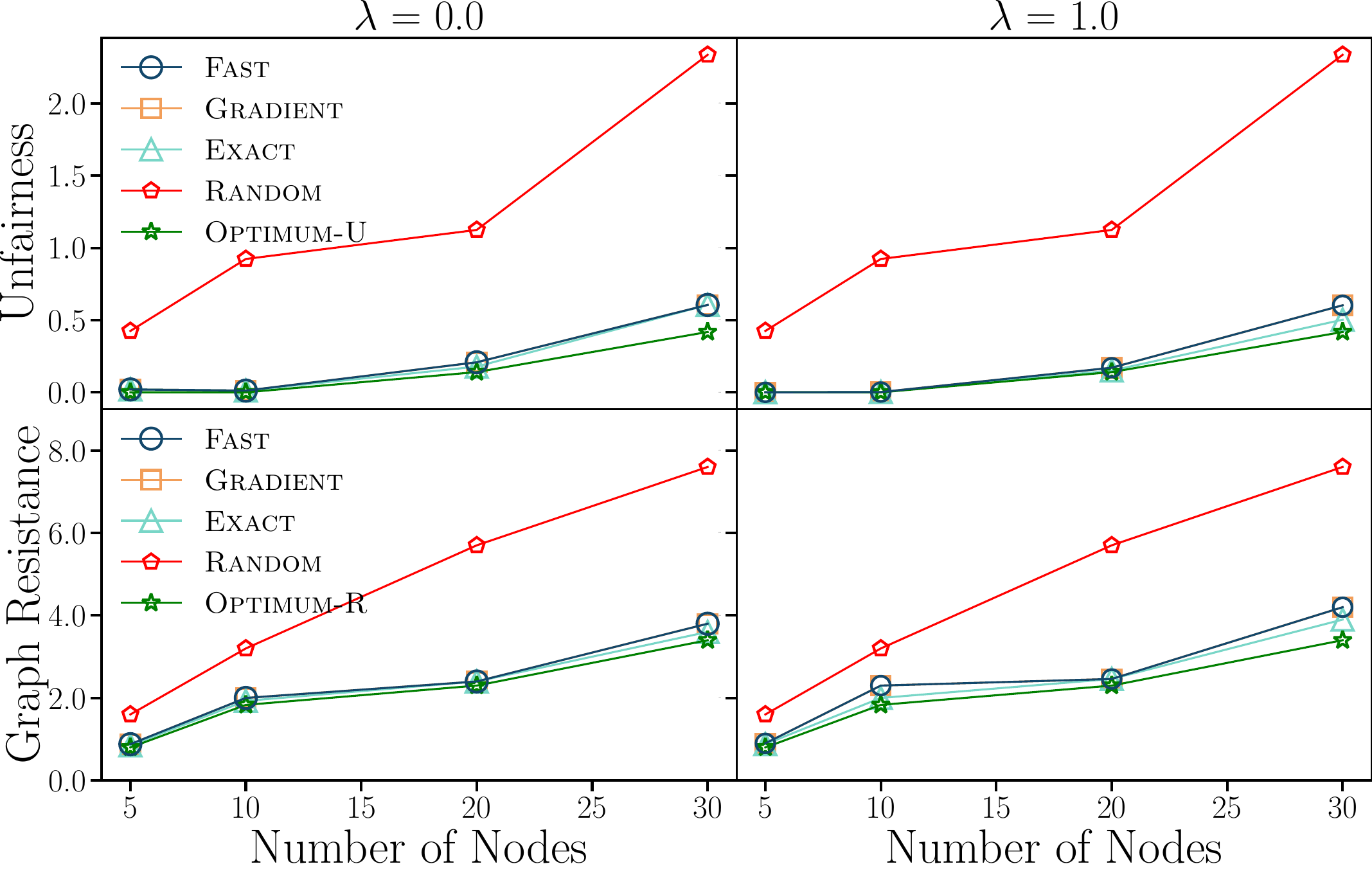}
    \vspace{-0.2cm}
    \caption{Comparison of the final unfairness and graph resistance for our algorithms, the random and optimal strategies.}
    \vspace{-0.2cm}
    \label{fig:opt}
\end{figure} 

{We compare our algorithms against baseline edge recommendation schemes—\textsc{DP}, \textsc{BP}, \textsc{DS}, \textsc{BS}, \textsc{ER}, and \textsc{Fiedler}—on larger real-world networks to further demonstrate their effectiveness. In all experiments, we greedily add 50 edges. When $\lambda=0$, the objective reduces solely to the graph resistance. In this case, we add \textsc{KirMin}~\cite{ZhAhZh25} as a specialized baseline and also include the method from~\cite{wang2014improving} without any adaptation, to preserve a faithful comparison. As shown in Fig.~\ref{fig:baseline1}, our approximation algorithms consistently achieve the largest reductions in the graph resistance across large graphs. Interestingly, \textsc{KirMin} obtains results close to ours under this special setting, but it is inherently restricted to the $\lambda=0$ case---once $\lambda \neq 0$, \textsc{KirMin} no longer applies. By contrast, our formulation flexibly generalizes to both single- and multi-objective optimization. When $\lambda=1$, the objective focuses exclusively on fairness. The results in Fig.~\ref{fig:baseline3} show that our algorithm consistently lowers unfairness more effectively than fairness-driven heuristics, which often yield only marginal improvements and may even increase disparity on some networks. Finally, when $\lambda=0.5$, the objective balances both goals. As illustrated in Fig.~\ref{fig:baseline2}, our greedy algorithms simultaneously reduce the graph resistance and unfairness more substantially and consistently than all baseline strategies, most of which improve only one metric in a limited or inconsistent manner. Taken together, these results highlight the effectiveness of our proposed algorithms across different $\lambda$ settings on large real-world graphs.}

\begin{figure}
    \centering
    \includegraphics[width=0.9\linewidth]{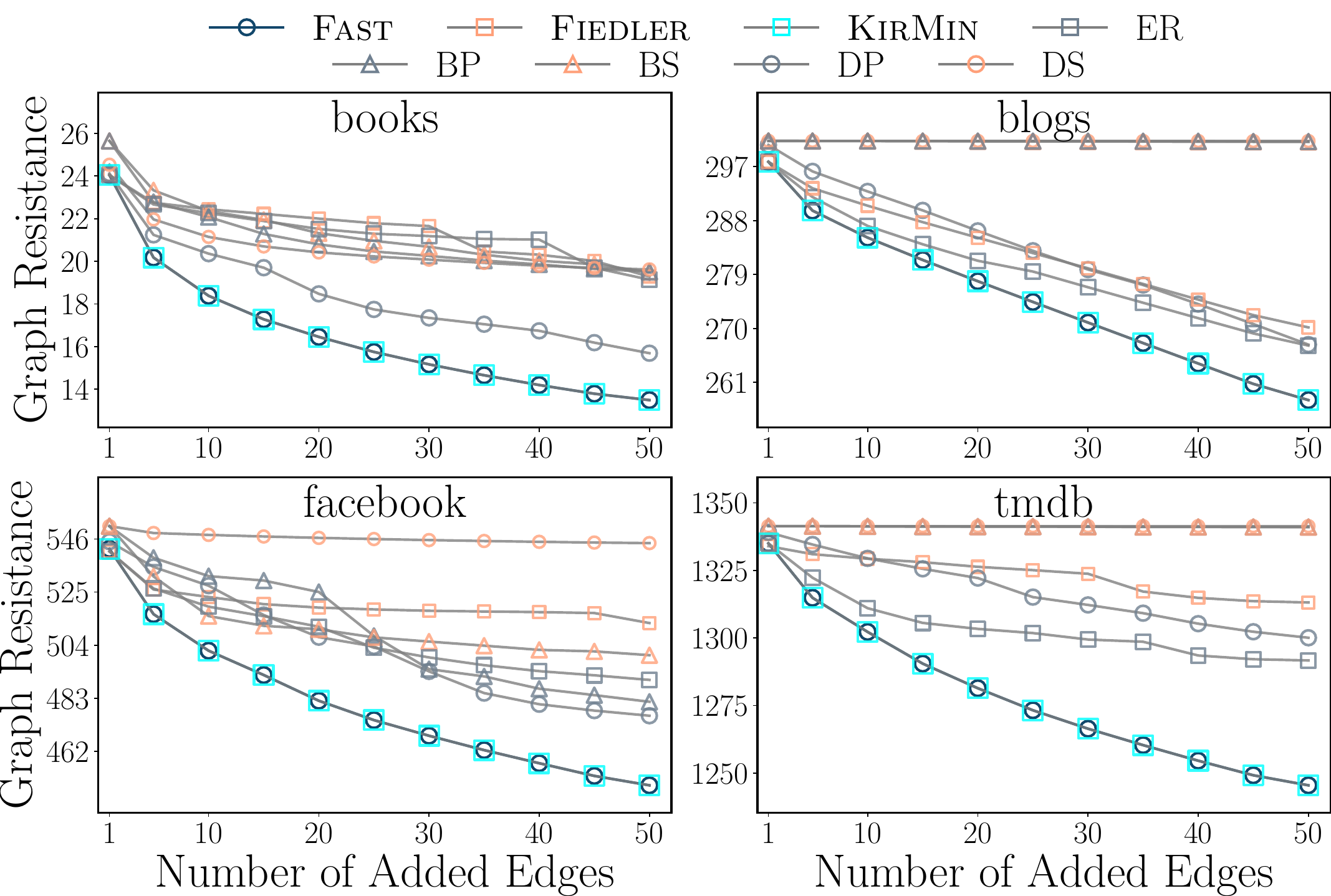}
    \vspace{-0.2cm}
    \caption{Graph resistance for different methods~($\lambda=0$).}
    \vspace{-0.3cm}
    \label{fig:baseline1}
\end{figure}
\begin{figure}
    \centering
    \includegraphics[width=0.9\linewidth]{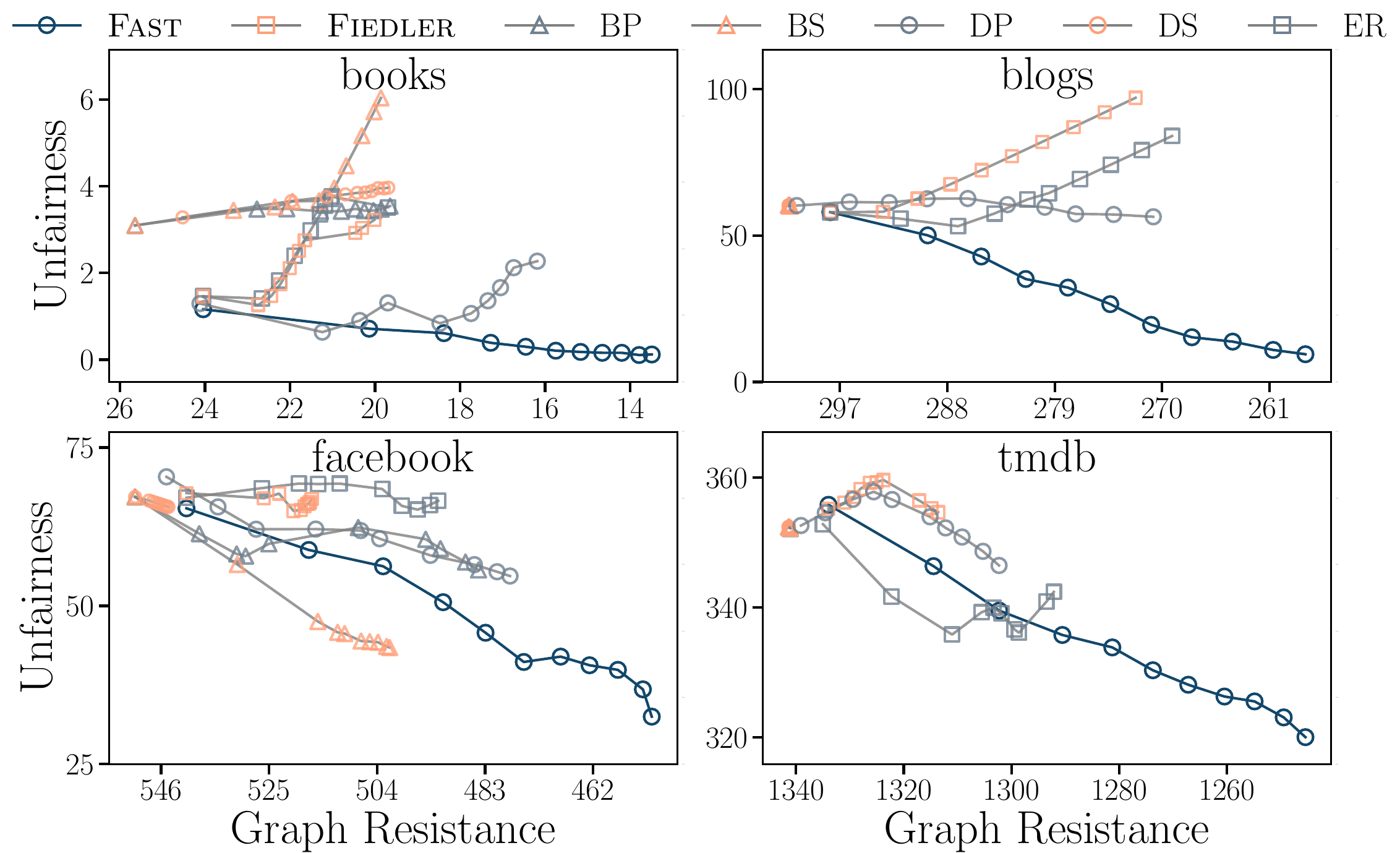}
    \vspace{-0.2cm}
    \caption{Graph resistance and unfairness for our algorithms and baselines~($\lambda=0.5$).}
    \label{fig:baseline2}
\end{figure}
\begin{figure}
    \centering
    \includegraphics[width=0.9\linewidth]{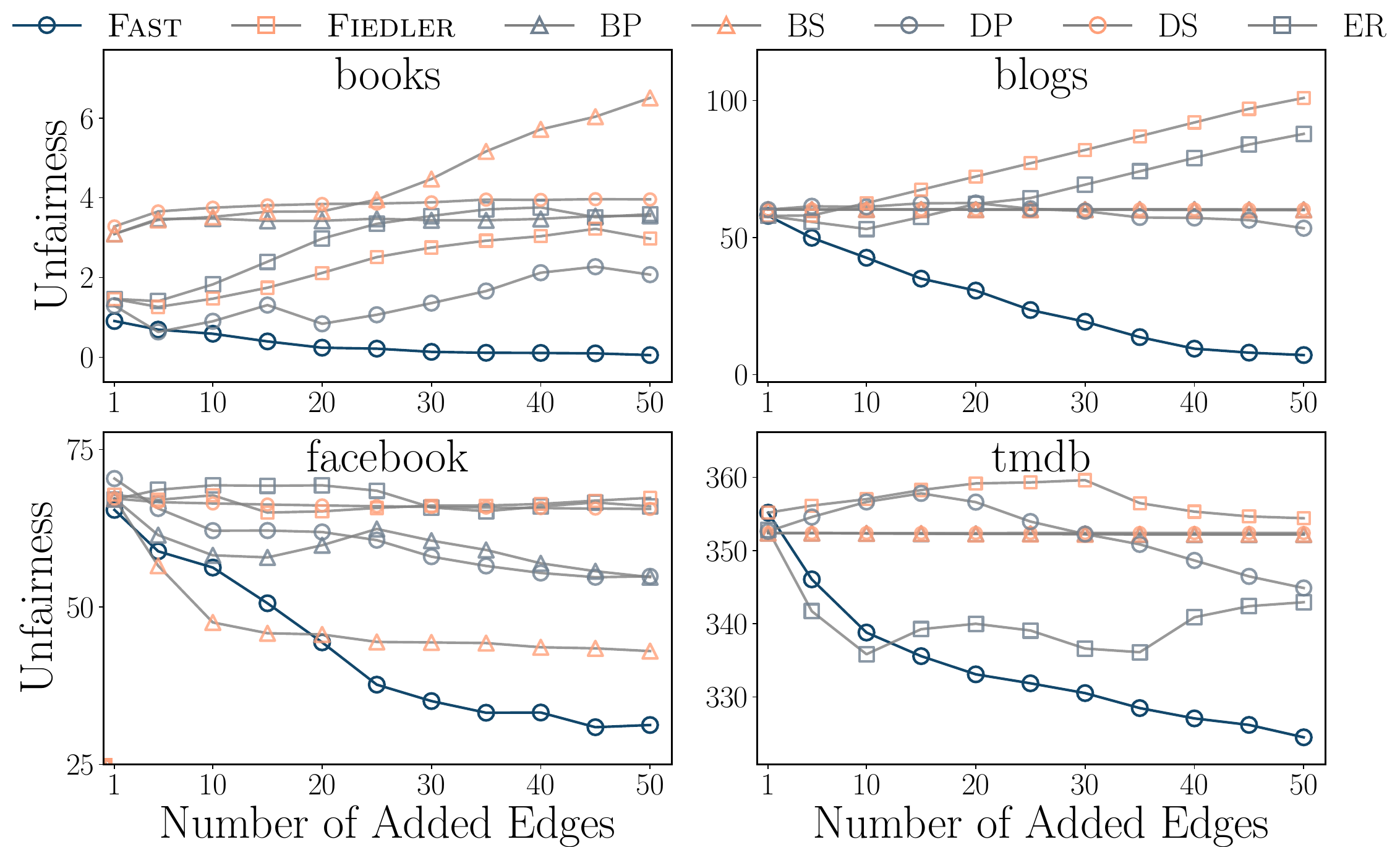}
    \vspace{-0.2cm}
    \caption{Unfairness for our algorithms and baselines~($\lambda=1$).}
    \label{fig:baseline3}
\end{figure}

\subsection{Comparison}
\begin{table*}
  \centering
  \caption{The running time (seconds, $s$) and the relative error of \Greedy and \Fast on real-world networks and synthetic networks of varying sizes for various $\epsilon$. For synthetic networks, the average and standard deviation are reported.}
  \vspace{-0.2cm}
  \label{tab:time}
  \fontsize{9}{9}\selectfont
  \setlength{\tabcolsep}{3pt}
  \begin{threeparttable}
    \begin{tabular}{llcccccccccc}
    \Xhline{2\arrayrulewidth}
    \multirow{2}*{Type} &
    \multirow{2}{*}{Network}&
    \multicolumn{4}{c}{Running time ($s$)} &
    \multicolumn{3}{c}{Relative error $\eta$  ($\times 10^{-2}$)} &
    \multicolumn{3}{c}{Relative error $\theta$ ($\times 10^{-2}$)} \\
    \cmidrule{3-6} \cmidrule{7-9} \cmidrule{10-12}
     && $\Greedy$ & $\eps{=}0.3$ & $\eps{=}0.2$ & $\eps{=}0.1$ 
     & $\eps{=}0.3$ & $\eps{=}0.2$ & $\eps{=}0.1$ 
     & $\eps{=}0.3$ & $\eps{=}0.2$ & $\eps{=}0.1$ \\
    \midrule
    \multirow{8}{*}{\rotatebox{90}{Real-world}}
    &books & 0.76 & 0.65 & 3 & 5 & 1.07 & 0.82 & 0.62 & 2.01 & 1.05 & 0.92 \\
    &blogs  & 187 & 9 & 19 & 72 & 1.30 & 0.75 & 0.13 & 5.07 & 3.75 & 0.61 \\
    &facebook  & 31,842 & 34 & 89 & 374 & 1.41 & 0.44 & 0.25 & 1.64 & 1.17 & 0.98 \\
    &tmdb  & 44,235 & 107 & 240 & 994 & 2.34 & 0.92 & 0.43 & 1.85 & 1.63 & 0.29 \\
    &dblp-course  & 212,151 & 213 & 445 & 1,849 & 3.21 & 2.21 & 0.49 & 4.90 & 2.31 & 0.91 \\
    &dblp-pub  & -- & 298 & 661 & 2,419 & -- & -- & -- & -- & -- & -- \\
    &dblp-gender & -- & 307 & 681 & 2,942 & -- & -- & -- & -- & -- & -- \\
    &dblp-aminer & -- & 3,121 & 7,021 & 23,211 & -- & -- & -- & -- & -- & -- \\
    &linkedin & -- & 16,401 & 30,244 & 151,231 & -- & -- & -- & -- & -- & -- \\
    \midrule
    \multirow{5}{*}{\rotatebox{90}{Synthetic}}
    &BAh100 
    & $8.2_{\pm 3.1}$ & $3.9_{\pm1.9}$ & $10.0_{\pm 4.9}$ & $23.1_{\pm7.4}$ 
    & $3.3_{\pm 0.6}$ & $0.9_{\pm0.3}$ & $0.6_{\pm0.3}$ 
    & $3.1_{\pm1.3}$ & $2.1_{\pm1.1}$ & $0.9_{\pm0.3}$ \\
    &BAh1000 
    & $120_{\pm 13.7}$ & $10_{\pm3.2}$ & $34_{\pm6.4}$ & $51_{\pm8.5}$ 
    & $2.1_{\pm1.2}$ & $0.5_{\pm0.3}$ & $0.4_{\pm0.2}$ 
    & $2.5_{\pm1.6}$ & $1.7_{\pm0.8}$ & $0.6_{\pm0.3}$ \\
    &BAh1E5 & -- & 119 & 254 & 1,129 & -- & -- & -- & -- & -- & -- \\
    &BAh1E6 & -- & 891 & 1,754 & 8,142 & -- & -- & -- & -- & -- & -- \\
    &BAh5E6  & -- & 32,014 & 74,721 & 271,465 & -- & -- & -- & -- & -- & -- \\
    \Xhline{2\arrayrulewidth}
    \end{tabular}
  \end{threeparttable}
\end{table*}
\vspace{-0.1cm}
\begin{table}[t]
  \centering
  \caption{Cardinality $c$ of the subset returned by \textsc{ApproxCH} for different $\epsilon$ across datasets.
  }
  \vspace{-0.2cm}
  \label{tab:c-values}
  \fontsize{9}{10}\selectfont
  \setlength{\tabcolsep}{5.7pt}
  \begin{threeparttable}
    \begin{tabular}{llccc}
      \Xhline{2\arrayrulewidth}
      \multirow{2}{*}{Type} & \multirow{2}{*}{Dataset} &
      \multicolumn{3}{c}{$c$ under different $\epsilon$} \\
      \cmidrule(lr){3-5}
      && $\epsilon{=}0.3$ & $\epsilon{=}0.2$ & $\epsilon{=}0.1$ \\
      \midrule
      \multirow{9}{*}{\rotatebox{90}{Real-world}}
        & books         & 9   & 16   & 23   \\
        & blogs         & 21   & 30   & 48  \\
        & facebook      & 36  & 61  & 120  \\
        & tmdb          & 40  & 76  & 138  \\
        & dblp-course   & 65  & 112  & 155  \\
        & dblp-pub      & 67  & 125  & 165  \\
        & dblp-gender   & 70  & 125  & 170  \\
        & dblp-aminer   & 113  & 210 & 260 \\
        & linkedin      & 220 & 330 & 410 \\
      \midrule
      \multirow{5}{*}{\rotatebox{90}{Synthetic}}
        & BAh100        & $12_{\pm 3}$   & $20_{\pm 2}$   & $30_{\pm 3}$   \\
        & BAh1000       & $20_{\pm 2}$   & $30_{\pm 4}$   & $40_{\pm 3}$ \\
        & BAh1E5        & $43_{\pm 4}$ & $80_{\pm 5}$ & $120_{\pm 4}$ \\
        & BAh1E6        & $125_{\pm 5}$ & $230_{\pm 5}$ & $280_{\pm 5}$ \\
        & BAh5E6        & $210_{\pm 6}$ & $310_{\pm 10}$ & $420_{\pm 10}$ \\
      \Xhline{2\arrayrulewidth}
    \end{tabular}
  \end{threeparttable}
  \vspace{-10pt}
\end{table}

As shown above, our algorithms outperform baseline algorithms in optimizing efficiency and fairness. Here, we focus on comparing \Greedy and \Fast against each other. We first compare their run-time on real-world networks and synthetic BAh networks of varying sizes, averaging results from 100 instances for synthetic networks. Our findings are given in Table~\ref{tab:time}. With $\lambda=0.5$, \Greedy fails to terminate within five days for networks over $13,000$ nodes, while \Fast provides solutions within two days for the linkedin network and one day for the BAh5E6 network (five million nodes). Therefore, \Fast is significantly faster than its deterministic counterpart.

Now, we compare the effectiveness of our algorithms in optimizing the objective function. Fig.~\ref{fig:balance} and~\ref{fig:opt} demonstrate that the solutions from \Greedy and \Fast are very similar. To quantify this, we compute the relative errors for $R$ and $U$. Let $\add$ and $\Tilde{E}_a$ be the edge sets from \Greedy and \Fast, respectively. We define $\eta=|R\left(\add\right)-R\big(\Tilde{E}_a\big)| / R\left(\add\right)$ and $\theta=|U\left(\add\right)-U\big(\Tilde{E}_a\big)| / U\left(\add\right)$. The relative errors, shown in Table~\ref{tab:time}, are small for all networks, with a maximum of 5.07\%. Thus, \Fast is not only faster, but also nearly as effective as \Greedy in optimizing the objective function.
\begin{figure}
    \centering
    \includegraphics[width=0.9\columnwidth]{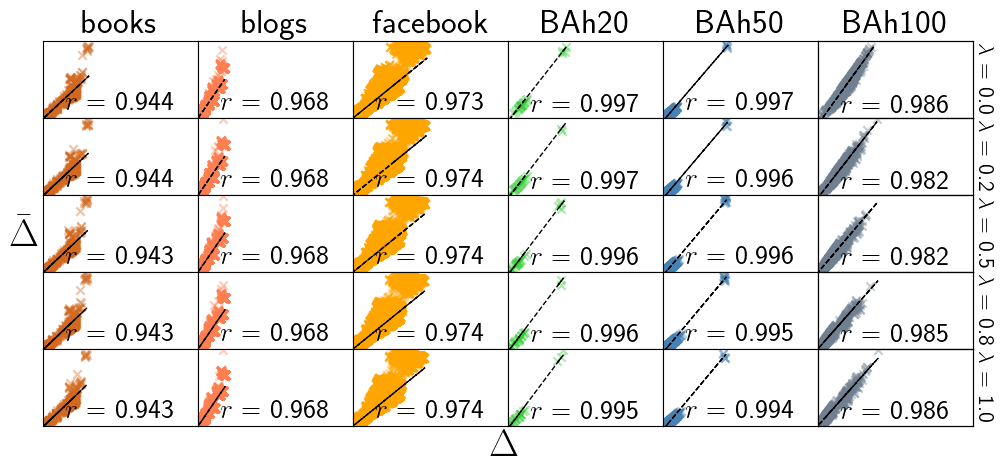}
    \vspace{-0.3cm}
    \caption{Correlation of $\Delta$ and $\bar{\Delta}$ for different types of networks with varying $\lambda$ values.
    }
    \vspace{-0.3cm}
    \label{fig:corre}
\end{figure}

It would be interesting to check to what extent the gradient computation in \Gradient affects its effectiveness compared to \Greedy. In Fig.~\ref{fig:corre}, we present the Pearson correlation coefficients to measure how correlated $\bar{\Delta}(e)$ is to the true decrease of $F$, $\Delta(e)$ on three real-world networks and three synthetic BAh networks, reporting average values for the synthetic cases. In all cases, $\bar{\Delta}(e)$ shows a high correlation (over 0.9) with the true 
 $\Delta(e)$, supporting the gradient computation in \Gradient.  Furthermore, this gives further credibility to the \Fast algorithm since \Fast is based on \Gradient and all other approximation techniques used in \Fast have theoretical guarantees.

To better illustrate the efficiency of \textsc{ApproxCH}, we also report the cardinality $c$ of the returned subset across datasets and different $\epsilon$. As shown in Table~\ref{tab:c-values}, $c$ decreases steadily as $\epsilon$ increases. Moreover, the values of $c$ remain manageable even on large datasets, explaining the practical scalability of our method.
   
\section{Conclusion and Future Work}
We fill a crucial gap in the field of network fairness by leveraging the concept of resistance distance as a measure of information access. We showed that it can be utilized to identify disadvantaged groups concerning information access. We then proposed effective algorithms, particularly a fast edge recommendation algorithm, for optimizing efficiency and mitigating bias in information access. We demonstrated through extensive experiments that our algorithm effectively optimizes fair information dissemination in large-scale social networks.

{A potential avenue for future research is to characterize the approximability of FIAM.} Additionally, further research is required to gain more insight into other measures of information access besides resistance distance and the trade-offs between different fairness criteria when mitigating bias in information access. { Since recommendations may be accepted probabilistically,  we plan to adapt the framework by recommending more than 
$k$ edges to ensure about $k$ are realized.}

Many networks emerging in real world are directed and weighted. A prospective research direction would be to shed some light on up to what extent our algorithmic machinery could be exploited to capture these larger classes of networks and what challenged are encountered that require development of novel approaches and techniques.
\section{Supporting Lemmas and Proofs}\label{sec:support}
\subsection{Supporting Lemmas for Section~\ref{sec:convex2}}
\begin{lemma}(\cite{AwKaZh20,kalantari2015characterization})\label{lem:cvx_hull}
Given a set \(P=\{\vvv_1,\vvv_2,\cdots,\vvv_n\}\) in \(\mathbb{R}^d\) and an error parameter \(\epsilon\in (0,1)\), the algorithm \(\hat{P}=\textsc{ApproxCH}(P,\epsilon)\) outputs a \(c\)-node subset \(\hat{P}\) of the node set $\bar{P}$ of the convex hull, running in \(\mathcal{O}(nc(d+\epsilon^{-2}))\) time, with any point \(p\in \Bar{P}\) having a distance to \(\hat{P}\) of at most \(\epsilon D_P\), where \(D_P\) is the diameter of \(P\).
\end{lemma}
\begin{lemma} \label{lem:jllemma}
(JL Lemma~\cite{johnson1984extensions}) Given fixed vectors $\vvv_1,\vvv_2,\cdots,\vvv_n\in \mathbb{R}^{d}$ and a real number $\epsilon > 0$. Let $q$ be a positive integer such that $q \geq 24 \log n / \epsilon^{2}$ and $\QQ_{q \times d}$ a random matrix with each entry being $1 / \sqrt{q} \text { or }-1 / \sqrt{q}$ with identical probability. Then, with probability at least $1-1/n$, the following statement holds for any pair of $1 \leq i, j \leq n$:
\begin{align}
    \notag (1-\epsilon)\|\vvv_{i}-\vvv_{j}\|^{2} \leq\|\QQ \vvv_{i}-\QQ \vvv_{j}\|^{2} \leq(1+\epsilon)\|\vvv_{i}-\vvv_{j}\|^{2}.
\end{align}
\end{lemma}
\begin{lemma} \label{lem:solver}(Fast SDDM Solver~\cite{SpTe14,CoKyMiPaJaPeRaXu14,solver2023}) There is a nearly linear-time solver $\mathbf{g} = \SDDMSolver(\LL, \bb, \delta)$ which takes a symmetric positive semi-definite matrix $\LL_{n\times n}$ with $m$ nonzero entries, a vector $\bb \in \mathbb{R}^n$, and an error parameter $\delta > 0$, and returns a vector $\mathbf{g} \in \mathbb{R}^n$ satisfying $\norm{\mathbf{g} - \LL^{-1} \bb}_{\LL} \leq \delta \norm{\LL^{-1} \bb}_{\LL}$ with probability at least $1-1/n$, where $\norm{\mathbf{g}}_{\LL} = \sqrt{\mathbf{g}^\top \LL \mathbf{g}}$. This solver runs in expected time $\tilde{\mathcal{O}}(m)$, where $\tilde{\mathcal{O}}(\cdot)$ notation suppresses the ${\rm poly} (\log n)$ factors. 
\end{lemma}
\begin{lemma} (\cite{li2019current}\label{lem:diags})
    Consider an error parameter $\epsilon\in(0,1)$ and a connected graph $\calG$ with $m$ edges, where $\LL^{\dag}$ is the pseudoinverse of its Laplacian matrix. There is a routine $\rr=\textsc{AppDiag}(\calG,\epsilon)$ that runs in $\tilde{\mathcal{O}}(m)$, and returns a vector $\rr$ containing all the diagonal elements of $\LL^{\dag}$, and for every node $u\in V$, with high probability it satisfies:
    \begin{align}
        (1-\epsilon)\ee_u^{\top} \LL^{\dagger} \ee_u\leq \rr_u\leq (1+\epsilon)\ee_u^{\top} \LL^{\dagger} \ee_u.
    \end{align}
\end{lemma}
\subsection{Supporting Proofs}
\noindent\textbf{Derivation of Definition~\ref{lem:delta_bar}}
If we add a small weight $\sigma$ to edge $e$, according to the  Sherman-Morrison formula~\cite{Me73}, $\LL^{\dag}$ will update to
\(
    (\LL+\sigma\bb_e\bb_e^\top)^\dag = \LL^\dag - \sigma\frac{\LL^\dag \bb_e \bb_e^\top \LL^\dag}{1+\sigma \bb_e^\top \LL^\dag \bb_e},
\)
and hence
\begin{align}
    \notag \frac{\partial{F}}{\partial{\ww_i}} =& \lim _{\sigma \rightarrow 0} \frac{1}{\sigma}\Big[(1-\lambda)\ldtrace{\LL^{\dag}-\MM}-(1-\lambda)\ldtrace{\LL^{\dag}}+\\
    \notag&\lambda \Big(\frac{n}{\abs S}\big(\LL^{\dag}-\MM\big)_{S}+\ldtrace{\LL^{\dag}-\MM}\Big)^2+\\
    \notag &\lambda\Big(\frac{n}{\abs T}\big(\LL^{\dag}-\MM\big)_{T}+\ldtrace{\LL^{\dag}-\MM}\Big)^2-\\
    \notag&\lambda\Big(\frac{n}{\abs S}\LL^\dag_{S}+ \ldtrace{\LL^\dag}\Big)^2-\lambda\Big(\frac{n}{\abs T}\LL^\dag_{T}+ \ldtrace{\LL^\dag}\Big)^2\Big]\\
    =&\notag-\bb_e^{\top}\Big[(1-\lambda)\LL^{2\dag}+\\
    \notag &2\lambda \frac{n}{\abs S}\Big(\ldtrace{\LL^{\dag}}+\frac{n}{\abs S}\LL^{\dag}_{S}\Big)\sum_{v\in S}\LL^{\dag}\ee_{v}\ee_{v}^{\top}\LL^{\dag}+\\
    \notag&2\lambda \frac{n}{\abs T}\Big(\ldtrace{\LL^{\dag}}+\frac{n}{\abs T}\LL^{\dag}_{T}\Big)\sum_{v\in T}\LL^{\dag}\ee_{v}\ee_{v}^{\top}\LL^{\dag}+\\
    &2\lambda\Big(2\ldtrace{\LL^{\dag}}+\frac{n}{\abs S}\LL^{\dag}_{S}+\frac{n}{\abs T}\LL^{\dag}_{T}\Big)\LL^{2\dag}\Big]\bb_e,\notag
\end{align}
where $\MM=\sigma \dfrac{\LL^{\dag} \bb_{e} \bb_{e}^{\top} \LL^{\dag}}{1+\sigma\bb_{e}^{\top} \LL^{\dag} \bb_{e}}$.
Based on the definitions of $\alpha$, $\beta$, and $\gamma$, we could conclude the proof.\qed

\noindent\textbf{Proof of Lemma~\ref{lem:furthest}}
Let $\xx, \yy= \argmax_{i,j\in P}\|\vvv_i-\vvv_j\|$ be the furthest pair of points in $P$, and $\xx^{\prime}, \yy^{\prime}= \argmax_{i,j\in \hat{P}}\|\vvv_i-\vvv_j\|$ be the furthest pair of points in $\hat{P}$. Based on Lemma~\ref{lem:cvx_hull}, we have $\|\xx-\xx^{\prime}\|\leq \epsilon/25 D_P$, and $\|\yy-\yy^{\prime}\|\leq \epsilon/25 D_P$. According to the triangle inequality,  
\begin{align}
    \notag\abs{D_P-D_{\hat{P}}}=&\abs{\|\xx-\yy\|-\|\xx^\prime-\yy^\prime\|}\leq \|\xx-\xx^\prime-\yy + \yy^\prime\|\\
    \leq &\|\xx-\xx^\prime\|+\|\yy-\yy^\prime\|\leq 2\epsilon/25 D_P\label{eq:cc1}.
\end{align}
Multiplying both sides of Eq.~\eqref{eq:cc1} by $D_P+D_{\hat{P}}$, we get 
\begin{align}\label{eq:cc2}
    \notag&\big|D_P^2-D_{\hat{P}}^2\big| =\big|\big(D_P+D_{\hat{P}}\big)\big(D_P-D_{\hat{P}}\big)\big|\\
    &\leq |D_P+D_{\hat{P}}||D_P-D_{\hat{P}}| \leq 2\epsilon/25 D_P\big(D_P+D_{\hat{P}}\big).
\end{align}
Dividing both sides of Eq.~\eqref{eq:cc2} by $D_P^2$ yields
\begin{align}
    \notag&\big|\big(D_P^2-D_{\hat{P}}^2\big)/D_P^2\big|\leq2\epsilon/25(1+D_{\hat{P}}/D_P)\leq\epsilon/5,
\end{align}
which concludes the proof.\qed

\noindent\textbf{Proof of Lemma~\ref{lem:error2}}	
To prove the lemma, it suffices to prove that for any node pair $e=(u,v)$,
	\begin{align*}
		\big|\|\ZZ \bb_e\|^2-\|\ZZtil \bb_e\|^2\big| 
		=&
		\big|\|\ZZ \bb_e\|-\|\ZZtil \bb_e\|\big|\times	\big|\|\ZZ \bb_e\|+\|\ZZtil \bb_e\|\big| \\\notag 
		\le&
		\left(\frac{2\epsilon}{125}+\frac{\epsilon^2}{15625}\right)\norm{\ZZ \bb_e}^2.
	\end{align*}

	The above equation can be further reduced to
		\begin{equation}\label{EE211}
		\notag \big|\norm{\ZZ \bb_e}-\|\ZZtil \bb_e\|\big| \le
		\frac{\epsilon}{125}\norm{\ZZ \bb_e}.
		\end{equation}

Using the triangle inequality, and letting $P_{uv}$ be a simple path connecting node $u$ and $v$, we get	
\begin{align*}
		\big|\norm{\ZZ \bb_e} -\|\ZZtil\bb_e\|\big|
		\leq \|(\ZZ - \ZZtil) \bb_e\| 
		\leq 
		\sum\limits_{(a, b) \in P_{uv}}
		\|(\ZZ - \ZZtil) \bb_{ab}\|.
\end{align*}

Using the Cauchy-Schwarz inequality, we obtain \begin{align*}
		&\bigg(\sum\limits_{(a, b) \in P_{uv}}
		\|(\ZZ - \ZZtil) \bb_{ab}\|\bigg)^2 
        \leq 
        n\sum\limits_{(a, b) \in P_{uv}}
		\|(\ZZ - \ZZtil) \bb_{ab}\|^2 \\
		\leq& 
        n\sum\limits_{(a, b) \in E}
		\|(\ZZ - \ZZtil) \bb_{ab}\|^2 
		=  n \| (\ZZ - \ZZtil) \BB^\top \|_F^2.
		\end{align*}

We transform the above Frobenius norm into $\LL$-norm:
\begin{align*}
		&n \|(\ZZ - \ZZtil) \BB^\top \|_F^2
		=n\ldtrace{(\ZZ - \ZZtil) \BB^\top \BB (\ZZ - \ZZtil)^\top} \\=&n\trace{(\ZZ - \ZZtil) \LL (\ZZ - \ZZtil)^\top} \\
		=&n\sum_{i=1}^q \norm{\zz_i - \tilde{\zz}_i}_{\LL}^2 
		\leq n\delta^2 \sum_{i=1}^q \norm{\zz_i}_{\LL}^2.
		\end{align*}

Considering $\ldtrace{\LL^{\dagger}} \le n(n^2-1)/6$ (cf.~\cite{LOVEJOY2003333}), we obtain
\begin{align*}
		&n\delta^2 \sum_{i=1}^q \|\zz_i\|_{\LL}^2 = n\delta^2
		\| \Bar{Q} \LL^{\dagger} \BB^\top \|_F^2 = n\delta^2
		\sum_{(a, b) \in E} \|\Bar{Q} \LL^{\dagger} \bb_{ab}\|^2
		\\
		\le &
		n\delta^2 \left(1+\frac{3\eps}{125}\right)\sum_{(a, b) \in E}
		\|\EE^{X}\LL^{\dagger} \bb_{ab}\|^2 \\
		=&n\delta^2\left(1+\frac{3\eps}{125}\right)
		\|\EE^{X}\LL^{\dagger} \BB^\top \|_F^2
		\\
		= & n\delta^2\left(1 + \frac{3\eps}{125}\right)\trace{\EE^{X}\LL^{\dagger}\BB^\top\BB \LL^{\dagger}\EE^{X}}
		\leq \delta^2 \frac{(1 + 3\eps/125)n^4}{6}.
		\end{align*}
On the other hand, we give a lower bound on $\norm{\ZZ \bb_e}^2$:
 \begin{equation}\label{EE22}
	\begin{split}
	\notag\norm{\ZZ \bb_e}^2
	&\geq \left(1-{3\eps}/{125}\right)\bb_e^{\top} \LL^{\dagger}\EE^{X}\LL^{\dagger} \bb_e \\
	&\geq \frac{1-3\epsilon/125}{\lambda_{n-1}^{2}} \geq \frac{1-3\epsilon/125}{n}.
	\end{split}
	\end{equation}
Combining the above two equations, we have
		\begin{align*}
		&\quad \frac{
			\big| \norm{\ZZ \bb_e} -  \|\ZZtil  \bb_e\|\big|^2
		}{
			\norm{\ZZ \bb_e}^2
		}
		\le
		\delta^2\frac{1+3\epsilon/125}{6(1-3\epsilon/125)} n^{5}
			\le \left(\frac{\epsilon}{125}\right)^2,
		\end{align*}
together with the initial condition, this ends the proof.\qed
\section*{Acknowledgment}
The work was supported by the National Natural Sci-
ence Foundation of China (Nos. 62372112 and 61872093).

\normalem

\ifCLASSOPTIONcaptionsoff
  \newpage
\fi

\end{document}

%% file: farthest.tex
\subsection{\Farthest: Diameter Search on an Approximate Hull}\label{sec:convex2}

To overcome the prohibitive cost of a brute-force search for the furthest node pair, we propose the \Farthest algorithm. The core strategy is to tackle the geometric equivalent of the problem—finding the diameter of the point set—through a series of carefully chosen approximations centered around the efficient estimation of the convex hull.

A key insight from geometry is that the diameter of a point set is determined by a pair of points on its convex hull. However, computing the exact convex hull has a prohibitive time complexity of $\mathcal{O}(n^{\lfloor d/2\rfloor})$ in high dimensions, making it impractical. To circumvent this, we leverage the \textsc{ApproxCH} algorithm~\cite{AwKaZh20,kalantari2015characterization}. Instead of finding the exact hull, this algorithm efficiently outputs a small subset of points, $\hat{P}$, that approximates the true convex hull. This step is crucial because it dramatically reduces the search space. We formalize the effectiveness of this approach in the following lemma, which guarantees that the diameter of this small approximate set is a faithful approximation of the true diameter.

\begin{lemma}\label{lem:furthest}
Let $\epsilon\in (0,1)$, $P$ be the input point set, $\hat{P}=\textsc{ApproxCH}(P,\epsilon/25)$, and $c=|\hat{P}|$. Define $D_P= \max_{i,j\in P}\|\vvv_i-\vvv_j\|$ and $D_{\hat{P}}= \max_{i,j\in \hat{P}}\|\vvv_i-\vvv_j\|$ as the diameters of $P$ and $\hat{P}$ respectively. Then, $D_P^2 \stackrel{\epsilon/5}{\approx}D_{\hat{P}}^2$.
\end{lemma}

\noindent\textbf{Remark.} The parameter \(c=|\hat{P}|\) depends on the network structure and the membership of groups \(S\) and \(T\). While it could be as large as \(n\), experiments in Section~\ref{experiments} show it is usually quite small for real-world social networks.

While this significantly narrows down the candidate pairs, two computational challenges remain.
First, the coordinate vectors exist in a very high-dimensional space (up to $d=2n$), which makes subsequent geometric computations like finding the approximate convex hull prohibitively expensive. To address this, we employ the Johnson-Lindenstrauss (JL) Lemma~\cite{johnson1984extensions}. This powerful result guarantees that a set of high-dimensional vectors can be projected into a much lower-dimensional space using a random matrix, while preserving all pairwise Euclidean distances up to a small relative error $(1\pm\epsilon)$. Specifically, it allows us to reduce the dimension from $d$ to a much more manageable $q = \mathcal{O}(\log n / \epsilon^2)$, which depends only logarithmically on the number of nodes. This step is critical for making the overall algorithm scalable.

Second, the initial computation of the coordinate vectors is itself a major bottleneck, as it requires solving linear systems with the graph Laplacian $\LL$ and calculating specific entries from its pseudoinverse $\LL^{\dagger}$—operations that are prohibitive if performed exactly. We accelerate this significantly by using two specialized techniques: a nearly linear-time solver for SDDM systems~\cite{SpTe14,CoKyMiPaJaPeRaXu14,solver2023} to efficiently approximate solutions to the linear equations, and the \textsc{AppDiag} routine~\cite{li2019current} to rapidly approximate the diagonal entries of $\LL^{\dagger}$ needed for the distance metric's scaling coefficients. These tools together bypass slow operations like direct matrix inversion and enable the fast construction of our coordinate vectors.

The application of these techniques leads to the construction of our approximate terms. Based on the definitions of $\alpha,\beta$, and $\gamma$, we obtain the corresponding $\epsilon$-approximations, namely $\tilde{\alpha}, \tilde{\beta}$ and $\tilde{\gamma}$. We then define the matrices for the approximate, low-dimensional coordinates. Let $q$ be a positive integer, and $\delta\in(0,1)$ be an error parameter. Let $\QQ_{1}$, $\QQ_{2}$, and $\QQ_{3}$ be three $q\times n$ random $\pm 1 / \sqrt{q}$ matrices. Let $\Bar{\QQ}_1=\QQ_1\EE^V$, $\Bar{\QQ}_2=\QQ_2\EE^S$, and $\Bar{\QQ}_3=\QQ_3\EE^{T}$. The core matrices are then computed using the fast solver: $\ZZtil_1[i,:] =\textsc{Solve}\left(\LL,\Bar{\QQ}_1[i,:]^{\top}, \delta\right)^{\top}$, $\ZZtil_2[i,:] =\textsc{Solve}\left(\LL,\Bar{\QQ}_2[i,:]^{\top}, \delta\right)^{\top}$, and $\ZZtil_3[i,:] =\textsc{Solve}\left(\LL,\Bar{\QQ}_3[i,:]^{\top}, \delta\right)^{\top}$. From these components, we construct the final approximate coordinate vector $\tilde{\cc}_i$ for node $i$:
\begin{align}
    \tilde{\cc}_i=\begin{pmatrix}\sqrt{\Tilde{\alpha}}\ZZtil_1^\top[i,:] & \sqrt{\Tilde{\beta}}\ZZtil_2^\top[i,:] &\sqrt{\Tilde{\gamma}}\ZZtil_3^\top[i,:]\end{pmatrix}^\top.
\end{align}

\noindent\textbf{Distance approximation.} Lemma~\ref{lem:app_delta} shows that the distance square $\Bar{\Delta}$ can be effectively approximated by our methods. As a preamble, we first present the following lemma, which is a key technical contribution of our work.

\begin{lemma}\label{lem:error2}
Given a graph $\mathcal{G}=(V, E)$ with $n$ nodes and an error parameter $\epsilon \in(0,1)$. Let $X \subseteq V$, $\QQ_{q \times n}$ be a random $\pm 1 / \sqrt{q}$ matrix where $q=\left\lceil\frac{24\log n}{(3\epsilon/125)^2}\right\rceil$, $\Bar{\QQ}=\QQ\EE^X$, $\ZZ_{q \times n} = \QQ \EE^X \LL^{\dagger}$, then for any node pair $e=(u,v)$:
	\[
	\big(1-\frac{3\eps}{125}\big) \|\EE^X\LL^{\dagger}\bb_{e}\|^2
	\leq
	\|\ZZ \bb_{e}\|^{2}
	\leq
	\big(1+\frac{3\eps}{125}\big) \|\EE^X\LL^{\dagger}\bb_{e}\|^2.
	\]
Let $\ZZtil[i,:] = \textsc{Solve}\left(\LL,\Bar{\QQ}[i,:]^{\top}, \delta\right)^{\top}$, $\zz_i$ be the $i$-th row of $\ZZ$, and let $\tilde{\zz}_i$ be an approximation of $\zz_i$ for $i \in \{1,2,...,q\}$, we have:
    \(
 	\|\zz_i-\tilde{\zz}_i\|_{\LL}\le\delta
	\norm{\zz_i}_{\LL},
	\)
	where
	$
	\delta \leq  \frac{\epsilon }{125}
	\sqrt{\frac{6(1-3\eps/125)}{n^{5}(1+3\eps/125)}}.
	$
Then, for every node pair $e=(u,v)$, we have:
    \begin{align}
    \label{EE201}
    \notag\big(1-\frac{\eps}{20}\big)\|\EE^{X}\LL^{\dag}\bb_e\|^2
    \leq 
    \|\ZZtil \bb_e\|^2 
    \leq 
    \big(1 + \frac{\eps}{20}\big) \|\EE^{X}\LL^{\dag}\bb_e\|^2,
    \end{align}
where $\ZZtil^\top = [\tilde{\zz}_1, \tilde{\zz}_2, ..., \tilde{\zz}_q]$.
\end{lemma}

Then, based on Lemma~\ref{lem:error2} and the multiplicative property of $\eps$-approximation {(Fact~\ref{fac:1})}, we can obtain the final guarantee for the total distance.
\begin{lemma}\label{lem:app_delta}
Let $\Tilde{\Delta}(e)=\Tilde{\alpha}\|\ZZtil_1 \bb_e\|^2+\Tilde{\beta}\|\ZZtil_2 \bb_e\|^2+\Tilde{\gamma}\|\ZZtil_3 \bb_e\|^2$, we have $\Tilde{\Delta}(e)\stackrel{\epsilon/5}{\approx} \Bar{\Delta}(e)$.
\end{lemma}

{The formal statements of the results from prior work (ApproxCH, JL Lemma, and AppDiag.) are provided in Section~\ref{sec:support}.} Building on the key elements proven above, we provide \Farthest (Algorithm~\ref{alg:furthest}) to find the best unconnected node pair. As illustrated in the pipeline (Fig.~\ref{fig:pipeline}), the algorithm begins by constructing a coordinate matrix $\XX$ with the help of the aforementioned techniques, then obtaining the set $P$ of points from $\XX$, followed by approximating the set $\hat{P}$ of points on the convex hull, and then traversing set $\hat{P}$ and identifying the locally optimal unconnected node pair.

\begin{algorithm}[t!]
	\caption{$\Farthest(\calG, \epsilon)$}\label{alg:furthest}
	\Input{
		A connected graph $\calG=(V,E)$ with $n$ nodes; an error parameter $\epsilon\in(0,1)$
	}
	\Output{
		A node pair $(i, j)$
	}
    $\delta_1 =  \frac{\epsilon }{125} \sqrt{\frac{6(1-3\eps/125)}{n^{5}(1+3\eps/125)}}$, $q=\left\lceil\frac{24\log n}{(3\epsilon/125)^2}\right\rceil$\;
    $\delta_2 =  \epsilon/20$, $\delta_3 =  \epsilon/25$\;

    Generate $q\times n$ random $\pm 1 / \sqrt{q}$ matrices $\QQ_1$, $\QQ_2$, and $\QQ_3$\;
    Compute $\Bar{\QQ}_1 = \QQ_1\EE^V$, $\Bar{\QQ}_2 = \QQ_2\EE^S$ and $\Bar{\QQ}_3 = \QQ_3\EE^T$\;
    Compute $\Tilde{\alpha}$, $\Tilde{\beta}$, and $\Tilde{\gamma}$ using elements in $\rr=\textsc{AppDiag}(\calG,\delta_2)$\;
	\For{$i = 1$ to $q$}{
		$\ZZtil_1[i,:] =\textsc{Solve}\left(\LL,\Bar{\QQ}_1[i,:]^{\top}, \delta_1\right)^{\top}$\;
        $\ZZtil_2[i,:] =\textsc{Solve}\left(\LL,\Bar{\QQ}_2[i,:]^{\top}, \delta_1\right)^{\top}$\; $\ZZtil_3[i,:] =\textsc{Solve}\left(\LL,\Bar{\QQ}_3[i,:]^{\top}, \delta_1\right)^{\top}$\;
	}
    Construct a matrix $\XX=(\sqrt{\Tilde{\alpha}}\ZZtil_1^\top \,\sqrt{\Tilde{\beta}}\ZZtil_2^\top \,\sqrt{\Tilde{\gamma}}\ZZtil_3^\top)^\top$\;
    $P = \{\XX[1,:], \XX[2,:], \cdots, \XX[n,:]\}$\;
    $\Hat{P}=\textsc{ApproxCH}(P, \delta_3)$\;
    Initialize $maxDistance \gets -\infty$\;
    Initialize $maxNodePair$ as an empty tuple\;
    \For{$i, j \in \Hat{P}$}{
        Compute $\Tilde{\Delta}(e=(i,j))$\;
            \If{$\Tilde{\Delta}(e) > maxDistance$}{
            $maxDistance \gets \Tilde{\Delta}(e)$\;
            $maxNodePair \gets (i, j)$\;}
    }
    \Return $maxNodePair$\;
\end{algorithm}

%% file: analysis.tex
{\subsection{Theoretical Analysis}}

The \Farthest algorithm (Algorithm~\ref{alg:furthest}) serves as the core subroutine for finding the most influential non-existent edge in each step. Its performance is summarized as follows.

\begin{lemma}
\label{lem:performance}
Let $\epsilon\in (0,1)$ and let $\Bar{\Delta}^{\max}=\max_{e\in\Bar{E}}\bar{\Delta}(e)$ and let $\Tilde{\Tilde{\Delta}}^{\max}$ be the output of the \Farthest algorithm. The time complexity of \Farthest (Algorithm~\ref{alg:furthest}) is $\tilde{\mathcal{O}}\left(\left(m+nc\right)\epsilon^{-2}\right)$, and its output satisfies $\Tilde{\Tilde{\Delta}}^{\max}\stackrel{\epsilon}{\approx} \Bar{\Delta}^{\max}.$
\end{lemma}

\begin{proof}
    For the time complexity, it takes $\tilde{\mathcal{O}}(m\eps^{-2})$ time to determine the node coordinates and $\tilde{\mathcal{O}}(nc\eps^{-2})$ time to approximate the convex hull. Therefore, the total time complexity of \Farthest is $\tilde{\mathcal{O}}\left(\left(m+nc\right)\epsilon^{-2}\right)$.
    For the approximation guarantee, let $\Tilde{\Delta}^{\max}=\max_{e\in\Bar{E}}\Tilde{\Delta}(e)$. From the results in Lemma~\ref{lem:furthest}, we have that the diameter found on the approximate hull is a good approximation of the diameter of the full set of approximate points, i.e., $\Tilde{\Tilde{\Delta}}^{\max}\stackrel{\epsilon/5}{\approx}\Tilde{\Delta}^{\max}$. Based on Lemma~\ref{lem:app_delta}, $\Tilde{\Delta}^{\max}\stackrel{\epsilon/5}{\approx} \Bar{\Delta}^{\max}$.
    Therefore,
    \begin{align}
        \notag (1-\eps)\Bar{\Delta}^{max}\le&(1 -\eps/5)^2\Bar{\Delta}^{max}\leq\Tilde{\Tilde{\Delta}}^{max}\\
        \le&(1+\eps/5)^2\Bar{\Delta}^{max}\leq(1+\eps)\Bar{\Delta}^{max}\notag,
    \end{align}
    which concludes the proof.
\end{proof}

Our main algorithm, \Fast (Algorithm~\ref{alg:fast}), iteratively calls \Farthest to select $k$ non-existent edges to add to the graph. The overall performance of \Fast is detailed as follows.

\begin{theorem}
    \label{the:performance}
    For any integer $0<k \le \bar{m}$ and any error parameter $\epsilon\in(0,1)$, \Fast runs in time $\tilde{\mathcal{O}}(k(m+nc)\epsilon^{-2})$ and outputs a set of $k$ edges. In each of the $k$ iterations, the algorithm selects a non-existing edge that is an $\epsilon$-approximation of the optimal choice at that step.
\end{theorem}

\begin{proof}
The approximation guarantee for each iteration is established by Lemma~\ref{lem:performance}. The \Fast algorithm performs $k$ rounds of this selection process (Lines 2-5 of Algorithm~\ref{alg:fast}). In each round, it calls \Farthest, taking $\tilde{\mathcal{O}}\left(\left(m+nc\right)\epsilon^{-2}\right)$ time. Consequently, the total time complexity of \Fast is $\tilde{\mathcal{O}}\left(k\left(m+nc\right)\epsilon^{-2}\right)$. As we will observe in Section~\ref{experiments}, $c$ is usually very small in real-world networks and $m=\mathcal{O}(n)$, which results in a nearly linear time complexity.
\end{proof}

{
\noindent\textbf{Remark}.
It is important to clarify the scope of guarantees offered by our surrogate-based approximation. 
The classical $(1-1/e)$ guarantees for cardinality-constrained submodular maximization do not apply as $F$ is monotone non-increasing but not submodular. 
This limitation is intrinsic to the problem structure rather than to our algorithm. 
Importantly, approximation error does not accumulate across iterations: each edge still decreases $F$, and the only effect of inexact marginals is that the decrease may be smaller than the maximum possible. Subsequent steps recompute marginals on the updated graph and do not inherit the earlier error. The experimental results in Section~VII provide strong evidence for this.
Moreover, when $F$ exhibits weak submodularity, a constant-factor greedy bound remains valid and degrades gracefully with estimation error.
}